\documentclass[a4paper]{article}
\usepackage[utf8]{inputenc}
\usepackage[T1]{fontenc}
\usepackage[top=3cm, bottom=3cm, left=3cm, right=3cm]{geometry}
\usepackage{fancyhdr}
\usepackage{microtype}
\usepackage{graphicx}

\usepackage{tikz}

\usepackage{scrextend}          %

\usepackage[textsize=scriptsize,disable]{todonotes}

\usepackage[citestyle=numeric,
            bibstyle=numeric,
            url=false,
            isbn=false,
            backref=true,
            backrefstyle=three,
            maxcitenames=3,
            maxbibnames=100,
            block=none,
            bibencoding=utf8,
            backend=biber]{biblatex}

\usepackage{hyperref}

\hypersetup{
  colorlinks=true,
  allcolors=blue,
  bookmarksnumbered=true,
  pdflang={en}, breaklinks}

\usepackage[fleqn]{amsmath}
\usepackage{amsfonts}
\usepackage{amssymb}
\usepackage{amsthm}
\usepackage{mathtools}
\usepackage{refcount}

\usepackage{paralist}
\usepackage{xspace}
\usepackage{doi}

\usepackage[capitalize,noabbrev]{cleveref}

\newcommand{\proofparagraph}[1]{\medskip\emph{#1}}
\newcommand{\proofparagraphns}[1]{\emph{#1}}

\newcommand{\meq}{\ensuremath{\ = \ }}
\newcommand{\mcoloneqq}{\ensuremath{\ \coloneqq \ }}

\newcommand{\mgeq}{\ensuremath{\ \geq \ }}

\newcommand{\fami}[1]{\ensuremath{\mathcal{#1}}}
\newcommand{\geom}[1]{\ensuremath{\mathfrak{#1}}}
\newcommand{\univ}[1]{\ensuremath{\mathbb{#1}}}
\newcommand{\comdots}{, \ldots, }

\newcommand{\ival}[2]{\ensuremath{\{#1\comdots #2\}}}
\newcommand{\ivals}[1]{\ival{1}{#1}}

\newcommand{\C}{\ensuremath{\mathcal{C}}\xspace}
\newcommand{\T}{\ensuremath{\mathcal{T}}\xspace}

\newcommand{\poly}{\ensuremath{\operatorname{poly}}}
\DeclareMathOperator{\glue}{\circ}

\newcommand{\no}{no}
\newcommand{\yes}{yes}

\newcommand{\noin}{\no-instance}
\newcommand{\yesin}{\yes-instance}

\newcommand{\cNP}{\ensuremath{\text{NP}}\xspace}

\newcommand{\fpt}{fixed-pa\-ram\-e\-ter tractable}

\newcommand{\fpty}{fixed-parameter tractability}
\newcommand{\wrt}{with respect to}

\newcommand{\Hyp}{\ensuremath{\mathcal{H}}}
\newcommand{\hyp}{\Hyp}
\newcommand{\Ver}{{V}}
\newcommand{\nEd}{\ensuremath{m}}
\newcommand{\nVer}{\ensuremath{n}}
\newcommand{\twnrl}{\ensuremath{\tau}}
\renewcommand{\S}{\ensuremath{\mathcal{S}}\xspace}
\newcommand{\G}{\ensuremath{\mathcal{G}}\xspace}

\newcommand{\HE}{\ensuremath{\fami{E}}} %
\newcommand{\Ed}{\HE}
\newcommand{\HV}{\ensuremath{V}} %

\newcommand{\layers}{\ensuremath{r}\xspace}
\newcommand{\lrs}{\layers}
\newcommand{\splane}{$\geom{S}$\nobreakdash-plane}
\newcommand{\scbd}{sphere-cut branch decomposition}
\newcommand{\loutp}{$\layers$\nobreakdash-out\-er\-pla\-nar}
\newcommand{\outp}[1]{\ensuremath{#1}\nobreakdash-out\-er\-pla\-nar}
\newcommand{\bdrd}[1]{\ensuremath{#1}\nobreakdash-bound\-aried}
\newcommand{\Oh}{\ensuremath{\operatorname{O}}}%

\newcommand{\PS}{\textsc{$\layers$-Outerplanar Support}}%
\newcommand{\pPS}{\textsc{$\layers$-Outerplanar Support}}%

\newcommand{\PSprobdef}{\problemdef{\PS}{A connected
    hypergraph $\Hyp$ with \(n\)~vertices and \(m\)~hyperedges, and \(\lrs\in\mathbb N\).}{Does $\Hyp$ admit an \loutp\ support?}}

\usepackage{xparse}

\newcounter{mppath}
\DeclareDocumentCommand\mppath{ o m }{%
   \addtocounter{mppath}{1}
   \def\fname{path\themppath.tmp}
   \input{\fname}
}

\usetikzlibrary{calc}

\pgfdeclarelayer{background}
\pgfsetlayers{background,main}

\tikzstyle {vert} = [circle, thick, fill=white, draw, inner sep=0pt, minimum size=2mm]
\newenvironment{mathenum}{%
  \begin{enumerate}[(i)]%
  }{%
  \end{enumerate}%
}
\newenvironment{lemenum}{\begin{compactenum}[(i)]}{\end{compactenum}}

\newtheorem{theorem}{Theorem}[section]
\newtheorem{rrule}[theorem]{Rule}

\newtheorem{definition}[theorem]{Definition}
\newtheorem{lemma}[theorem]{Lemma}
\newcommand{\prob}[6]{%
  \begin{quote}
    \begin{labeling}{#6}%
      \setlength\topsep{-.6ex} \setlength\itemsep{-.8ex}
    \item[#1]
    \item[\emph{#2}]#3
    \item[\emph{#4}]#5
    \end{labeling}%
  \end{quote}%
}

\newcommand{\problemdef}[3]{\prob{#1}{Input:}{#2}{Question:}{#3}{as}}
\crefname{property}{Property}{Properties}
\creflabelformat{property}{(#2#1#3)}

\crefname{stat}{Statement}{Statements}
\creflabelformat{stat}{(#2#1#3)}

\crefname{rrule}{Rule}{Rules}
\newcommand{\citet}[2][]{\textcite[][#1]{#2}}

\title{The role of twins in computing planar supports of hypergraphs%
  \thanks{An extended abstract of this work appeared at GD'16
    \cite{BKK+17}.
    This version provides full proof details, more illustrative figures, and a simpler example for the necessity of twins with, additionally, smallest-possible hyperedge size.}
}
\author{René van Bevern\thanks{Results reported in this article are unrelated to the work at Huawei.}\\
  \small Huawei Technologies Co., Ltd.\\
  \small \texttt{rene.van.bevern@huawei.com}
  \and
  Iyad Kanj\\
  \small DePaul University,\\
  \small Chicago, USA\\
  \small \texttt{ikanj@cs.depaul.edu}
  \and
  Christian Komusiewicz\\
  \small Phillips-Universität Marburg,\\
  \small Marburg, Germany\\
  \small \texttt{komusiewicz@informatik.uni-marburg.de}
  \and
  Rolf~Niedermeier\\
  \small Algorithmics and Computational Complexity\\
  \small Faculty~IV, TU~Berlin, Berlin, Germany\\
  \small \texttt{rolf.niedermeier@tu-berlin.de}
  \and
  Manuel Sorge\\
  \small Institute of Logic and Computation\\
  \small TU Wien, Vienna, Austria\\
  \small \texttt{manuel.sorge@ac.tuwien.ac.at}}

\begin{document}

\maketitle

\begin{abstract}
  \looseness=-1
  A \emph{support} or \emph{realization}
  of a hypergraph~$\Hyp{}$
  is a graph~\(G\) on the same vertex set as~\(\Hyp\) such that
  for each hyperedge of~$\Hyp{}$ it holds that its vertices induce a 
  connected subgraph of~$G$.
  The NP-hard problem of finding a \emph{planar} support
  has applications in hypergraph drawing
  and network design.
  Previous algorithms for the problem assume that
  \emph{twins}---pairs of vertices that are in precisely
  the same hyperedges---can safely be removed from the input hypergraph.
  We prove that this assumption is generally wrong,
  yet that the number of twins
  necessary for a hypergraph %
  to have a planar support only depends on its number of hyperedges.
  We give an explicit upper bound on the number of twins
  necessary
  for a hypergraph with \(m\)~hyperedges
  to have an \(r\)-outerplanar support,
  which depends only on~\(r\) and~\(m\).
  Since all additional twins can be safely removed,
  we obtain a linear-time
  algorithm for computing \(r\)-outerplanar supports
  for hypergraphs with \(m\)~hyperedges
  if $m$ and $r$
   are constant; in other words, the problem is fixed-parameter linear-time
   solvable with respect to the parameters $m$ and~$r$.

 \end{abstract}

\newcommand{\sepnum}{\ensuremath{2^{m \cdot (2\layers^2 + \layers + 1)}}}
\newcommand{\signumhalved}{\ensuremath{(r + 1)^{32r^2 + 8r}}}
\newcommand{\signumbits}{\ensuremath{(32r^2 + 8r)\cdot \log(r + 1) + 1}}

\newcommand{\suppsize}{\ensuremath{2^{6r \cdot 2^{\nEd{} \cdot (2\lrs^2 + \lrs + 1)} \cdot \signumhalved}}}

\section{Introduction}\label{sp:sec:intro}
Hypergraph drawings are useful as visual aid in diverse applications~\cite{AMAHMR16}, among them electronic circuit design~\cite{EGB06,FLR97} and relational databases~\cite{Mak90,BFMY83}.
This led to several generalizations of the concept of planarity
from graphs to hypergraphs.
The earliest among them is the attempt of Zykov~\cite{Zyk74},
who defined a hypergraph to be planar if its incidence graph is; this is equivalent
  to the requirement that one can draw the hyperedges
  as closed regions
  in such a way that each intersection of hyperedges
  contains exactly one vertex \cite{Zyk74}.
\citet{VF84} introduced \emph{planar realizations} (for an English reproduction of the results refer to the book of \citet{FLR97})
which
nowadays are better known as \emph{planar supports}~\cite{JP87,KKS08}:
a \emph{support} for a hypergraph~$\Hyp = (\HV, \HE)$ is a graph~$G$ on the same vertex set as~$\Hyp$ such that each hyperedge~$F \in \HE$ induces a connected subgraph~$G[F]$.
This is a generalization of planarity:
an ordinary graph is planar if and only
if it has a planar support when viewed as a hypergraph.

\looseness=-1
We study the NP-complete \cite{AS87,JP87}
problem of recognizing
hypergraphs that allow for a planar support.
These
are exactly the hypergraphs allowing for a
\emph{subdivision drawing}~\cite{JP87,KKS08}: given a hypergrap~$\Hyp$, we divide the plane into closed regions that one-to-one correspond to the vertices of~$\Hyp$ in such a way that, for each hyperedge~$F$, the union of the regions corresponding to the vertices in~$F$ is connected. Subdivision drawings have also been called \emph{vertex-based Venn diagrams}~\cite{JP87}.
\cref{fig:subset-standard-vs-subdivision} shows an example for such a drawing.
\begin{figure}[t!]
  \centering
  \hfill
  \includegraphics{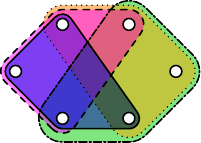}
  \hfill
  \includegraphics{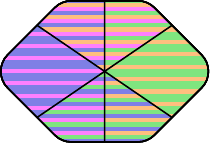}
  \hfill\mbox{}
  \caption{Two drawings of the same hypergraph. On the left, we see a
    drawing in the \emph{subset standard} in which the vertices (white
    circles) are enclosed by curves that correspond to hyperedges. On
    the right, we see a \emph{subdivision drawing} in which we assign
    vertices to regions (enclosed by black lines) and we color these
    regions with colors that one-to-one correspond to the hyperedges;
    for each hyperedge, the union of the regions of the vertices in that hyperedge
    is connected.}%
  \label{fig:subset-standard-vs-subdivision}
\end{figure}

\looseness=-1
Having a planar support, or equivalently, a subdivision drawing,
is a rather general concept of planar embeddings:
for example, each hypergraph that
has a planar incidence graph
or a well-formed Euler diagram \cite{FFH08}
has a planar support.
Still,
in the same way that most ordinary graphs are not planar,
most hypergraphs do not have planar supports \cite{FLR97}.
Actually \emph{finding} them might be even more complicated
by the fact that
several works on planar supports assume that the input hypergraph is twinless,
that is, there are no two vertices contained in precisely the same hyperedges (see Mäkinen~\cite[p.~179]{Mak90}, \citet[p.~535]{BKMSV11}, and \citet[p.~399]{KKS08}). Twins do not seem useful at first glance: whatever role one vertex can play to obtain a  planar support,
its twin can also fulfill. One of our contributions is disproving the
general validity of this assumption in \cref{sec:twins-planar}.  More specifically, we give a hypergraph with two twins that has a planar support but after removing one twin it ceases to have one.
Thus, twins may be crucial
to allow for a planar support.

\looseness=-1
More generally, we can construct hypergraphs with \(\ell\)~twins that allow for a planar support but cease doing so when removing one of the twins.  However, the number of hyperedges in the construction grows with~$\ell$. It is thus natural to ask whether there is a function~$\psi\colon \mathbb{N} \to \mathbb{N}$ such that, in each hypergraph with $m$~hyperedges, we can forget all but $\psi(m)$~twins while maintaining the property of having a planar support. Using well-quasi orderings, one can prove the \emph{existence} of such a function~$\psi$ (see \cref{sec:nonuni-fpt}), yet finding a closed form for~$\psi$ turned out to be surprisingly difficult: so far we could only compute a concrete upper bound when considering a second parameter, the outerplanarity number~$r$ of the desired planar support.
A graph is \emph{\loutp} if it admits a planar embedding (without edge crossings) which has the property that, after $\lrs$ times of removing all vertices on the outer face, we obtain an empty graph.
The outerplanarity number~$r$ of the support roughly translates to the number of layers in a corresponding drawing which can be seen by examining the construction of a subdivision drawing from a support given by Kaufmann, Kreveld, and Speckmann~\cite[p.~401]{KKS08} for example.
Formally, we study the following problem (special cases of which were also considered previously~\cite{BCPS11,BKMSV11}):

\PSprobdef
\noindent Herein, a hypergraph is connected if for every pair of vertices $u, v$ there is an alternating sequence of vertices and hyperedges that begins with $u$ and ends with $v$ such that successive elements are incident with each other. This assumption helps avoiding edge cases; our results easily extend to the non-connected case.

Our main result is a concrete upper bound on the number~\(\psi(m,r)\) of twins that might be necessary to obtain an \(r\)-outerplanar support.  Since superfluous twins can then be removed in linear time, this gives the following algorithmic result.

\begin{theorem}\label{thm:lintime}
  There is an algorithm solving \PS{} which, for constant~\(r\) and \(m\), has linear running time.
\end{theorem}
\noindent
To put \cref{thm:lintime} into perspective, \PS{} remains \cNP-complete for $r = \infty$~\cite{AS87,JP87} and even for every fixed $r > 1$ \cite{BKMSV11} (see below).
The constants in the running time of the algorithm in \cref{thm:lintime} have a large dependence on~$m$ and~$r$.
However, as also the discussion about twins above shows, the number of hyperedges is a natural parameter whose influence on the complexity is interesting to know.
Furthermore, it is conceivable that the parameters~$m$ and~$r$ are small in practical instances: for a large number~$m$ of hyperedges, it is plausible that we obtain only hardly legible drawings unless the hyperedges adhere to some special structure, whereas every hypergraph with at most eight hyperedges has a planar support \cite{FLR97,Pil86}.
Thus, it makes sense to design algorithms particularly for hypergraphs with a small number of hyperedges, as done by~\textcite{VV04,HKvK+18}.
Moreover, %
a small outerplanarity number~$r$ leads to few layers in the drawing which may lead to aesthetically pleasing drawings, similarly to path- or cycle-supports~\cite{BKMSV11}.

\paragraph{Related work.} For specifics on the relations of different planar hypergraph embeddings, see \citet{BCPS11},~\citet{FLR97}, and Kaufmann, Kreveld, and Speckmann~\cite{KKS08}.

\citet{AS87,JP87} showed that finding a planar support is \cNP-complete. \citet{BKMSV11} proved that \pPS\ is \cNP-complete for~$r = 2, 3$. {F}rom their proof
it follows that \pPS\ is also \cNP-complete for every $r > 3$. This is due to a property of the reduction that \citeauthor{BKMSV11} use: Given a formula~$\phi$ in 3CNF, they construct a hypergraph~$\Hyp$ that has a \emph{planar} support if and only if $\phi$ is satisfiable. Due to the way in which $\Hyp$ is constructed, if there is any planar support, then it is \outp{3}. Thus, deciding whether
there is an \loutp\ support for $r \geq 3$
also decides the satisfiability of the corresponding formula.

\looseness=-1
Towards determining the computational complexity of finding an outerplanar hypergraph support, \citet{BCPS11} gave a polynomial-time algorithm for cactus supports (graphs in which each edge is contained in at most one cycle). They also showed that finding an outerplanar support (or planar support) can be done in polynomial time if, in the input hypergraph, each intersection or difference of two hyperedges is empty, a singleton, or again a hyperedge in the hypergraph.
A tree support can even be found in linear time~\cite{BFMY83,TY84},
a tree support of minimum diameter can be found in polynomial time~\cite{Mel97},
and
one can deal with an additional upper bound on the vertex degrees in the tree support in polynomial time~\cite {BKMSV11}.
\citet{KMN14}
studied so-called area-proportional Euler diagrams, for which the corresponding computational problem reduces to finding a minimum-weight tree support. Such supports can also be found in polynomial time~\cite{KMN14,KS03}.
Furthermore, if there are only two hyperedges and the positions of the vertices in the embedding are specified, then checking for a planar support can be done in polynomial time~\cite{HKvK+18} 

\looseness=-1
In a wider scope, motivated by drawing metro maps and metro map-like diagrams, \citet{BCPS12} studied the problem of finding \emph{path-based} planar hypergraph supports (these are planar supports that fulfill the additional constraint that the subgraph induced by each hyperedge contains a Hamiltonian path) giving \cNP-hardness and tractability results. %
Path-based tree supports can also be found in polynomial time~\cite{SW94}.

A concept related to subdivision drawings is (overlapping) clustered planarity \cite{DGL07,AC17}.
Very roughly, a graph~$G$ together with a hypergraph~$\Hyp$ on the same vertex set is \emph{overlapping clustered planar} if $G$ and $\Hyp$ admit a joint embedding in the plane which is edge-crossing-free for~$G$, a subdivision drawing for~$\Hyp$, and no edge of~$G$ crosses twice the boundary of a region corresponding to a hyperedge in~$\Hyp$.
Overlapping clustered planarity is a generalization of clustered planarity;
In the latter,
one assumes that each pair of hyperedges in the hypergraph~$\Hyp$ is either disjoint or one hyperedge is a subset of the other~\cite{FCE95}.
Clustered planarity has attracted a lot of research interest, see Da~Lozzo et al.~\cite{LEGG19} for a recent overview of the literature.
Only recently, a polynomial-time algorithm for testing clustered planarity has been found~\cite{FT19}.

\citet{VF84} suggested a data reduction rule for finding planar supports
that removes more than just twins:
it keeps only one vertex out of
each nonempty inclusion-minimal intersection
of any number of hyperedges.
They prove that the input hypergraph has a planar support
if the reduced hypergraph has,
but that the reverse direction does generally not hold
(\citet{FLR97} show an example).

\citet{CKNSSW15} showed that for obtaining minimum-edge supports (not
necessarily planar), twins show a similar behavior as for \loutp\
supports: Removing a twin can increase the minimum number of edges
needed for a support and finding a minimum-edge support is linear-time
solvable for a constant number of hyperedges via removing superfluous
twins.

\paragraph{Organization.}
In \cref{sec:prelim} we provide some technical preliminaries used
throughout the work.
In \cref{sec:twins-planar} we
give an example that shows that twins can be crucial for a hypergraph
to have a planar support.
As mentioned, for each~$m \in \mathbb{N}$, there is a number
$\psi(m)$ such that in each hypergraph with a planar support we can
safely forget all but $\psi(m)$ twins (see \cref{sec:nonuni-fpt}).
In \cref{sec:application} we
give a concrete upper bound for $\psi(m)$ in the case of \loutp\ supports and
derive the linear-time algorithm for \PS{} claimed in \cref{thm:lintime}. We base the proof on a construction of a special sequence of nested separators in $r$-outerplanar graphs which is given in \cref{sec:sepseq}. We conclude and give some directions for future research in \cref{sec:concl}.

\section{Preliminaries}\label{sec:sparse-planar-prelim}\label{sec:prelim}

By $A \uplus B$ we denote the union of two disjoint sets~$A$ and~$B$. For a family of
sets~$\fami{F}$, we write $\bigcup\fami{F}$ in place of $\bigcup_{S
  \in \fami{F}}S$. For equivalence relations~$\rho$ over some set~$S$
  and $v\in S$,
  we use~$[v]_\rho$ to denote the equivalence class of~$v$ in~$\rho$.

\paragraph{Hypergraphs.}
A \emph{hypergraph} $\Hyp$ is a pair~$(\Ver,\Ed)$ consisting of a \emph{vertex set} $\Ver$, also denoted $V(\Hyp)$, and a \emph{hyperedge set}~$\Ed$, also denoted~$\Ed(\Hyp)$. The hyperedge set $\Ed$ is a family of subsets of $\Ver$, that is, $F \subseteq \Ver$ for every hyperedge~$F \in \Ed$. Where it is not ambiguous, we use $\nVer{}\coloneqq |\Ver|$ and $\nEd{}\coloneqq |\Ed|$. When specifying running times, we use $|\Hyp|$ to denote $|\Ver(\Hyp)| + \sum_{F \in \HE(\Hyp)}|F|$. The \emph{size}~$|F$| of a hyperedge~$F$ is the number of vertices in it. Unless stated otherwise, we assume that hypergraphs do not contain hyperedges of size at most one or multiple copies of the same hyperedge. (These do not play any role for the problem under consideration, and removing them can be done easily and efficiently.)

A vertex~$v \in V$ and a hyperedge~$F \in \HE$ are \emph{incident} with one another if~$v \in F$. For a vertex~$v \in \Ver(\Hyp)$, let $\Ed_{\Hyp}(v)\coloneqq \{F \in {\cal H} \mid v \in F\}$. If it is not ambiguous, then we omit the subscript~$\Hyp$ from $\Ed_{\Hyp}$. A~vertex~$v$ \emph{covers} a vertex~$u$ if $\Ed(u) \subseteq \Ed(v)$. Two vertices $u, v \in \Ver$ are \emph{twins} if $\Ed(v)=\Ed(u)$. Clearly, the relation~$\twnrl$ on~$\Ver$ defined by $\forall u, v \in \Ver \colon (u, v) \in \twnrl \iff\Ed(u)=\Ed(v)$ is an equivalence relation.
The equivalence classes $[u]_{\twnrl}$ for $u \in \Ver$ are called \emph{twin classes}.

\emph{Removing a vertex subset}~$S \subseteq \Ver(\Hyp)$ from a hypergraph $\Hyp = (\Ver,\Ed)$ results in the hypergraph~$\Hyp-S\coloneqq (\Ver\setminus S, \Ed')$, where~$\Ed'$ is obtained from $\{F\setminus S\mid F\in \Ed\}$ by removing empty and singleton sets. For brevity, we also write $\Hyp-v$ instead of $\Hyp-\{v\}$. The hypergraph $\Hyp$ \emph{shrunken to} $\HV' \subseteq \HV$ is the hypergraph~$\Hyp|_{\HV'} \coloneqq \Hyp-(\Ver\setminus \HV')$.

\paragraph{Graphs.} %
Our notation related to graphs is basically standard and heavily borrows from \citeauthor{Die16}'s book~\cite{Die16}. In particular, a \emph{bridge} of a graph is an edge whose removal increases the graph's number of connected components. Analogously, a \emph{cut-vertex} is a vertex whose removal increases the graph's number of connected components. Some special notation including the \emph{gluing} of graphs is given below. %

\paragraph{Boundaried graphs and gluing.}
For a nonnegative integer~$b \in \mathbb{N}$, a \emph{\bdrd{b}\ graph} is a triple~$(G, B, \beta)$, where $G$ is a graph, $B \subseteq V(G)$ such that $|B| = b$, and $\beta \colon B \to \ival{1}{b}$ is a bijection. Vertex subset~$B$ is called the \emph{boundary} and $\beta$ the \emph{boundary labeling}. For ease of notation we also refer to~$(G, B, \beta)$ as the \emph{\bdrd{b}} graph~$G$ with boundary $B$ and boundary labeling~$\beta$. For brevity, a \bdrd{b}\ graph~$G$ whose boundary is the domain of~$\beta$ and whose boundary labeling is~$\beta$ is also called \emph{\bdrd{\beta}}.

For a nonnegative integer~$b$, the \emph{gluing} operation $\glue_b$ maps two \bdrd{b} graphs to an ordinary graph as follows: Given two \bdrd{b}\ graphs $G_1, G_2$ with corresponding boundaries~$B_1, B_2$ and boundary labelings~$\beta_1, \beta_2$, to obtain the graph $G_1 \glue_b G_2$ take the disjoint union of $G_1$ and $G_2$, and identify each pair~$v \in B_1$ and $w\in B_2$ of vertices such that $\beta_1(v)=\beta_2(w)$. We omit the index $b$ in~$\glue_b$ if it is clear from the context.

\paragraph{Topology.}
A \emph{topological space} is a pair~$\geom{X} = (X, \fami{F})$ of a set $X$, called \emph{universe}, and a collection~$\fami{F}$ of subsets of $X$, called \emph{topology}, that satisfy
the following properties:
\begin{compactitem}
\item The empty set~$\emptyset$ and $X$ are in $\fami{F}$.
\item The union of the sets of any subcollection
  of $\fami{F}$ is in $\fami{F}$.
\item The intersection of the sets of any finite subcollection
  of $\fami{F}$ is in $\fami{F}$.
\end{compactitem}
\newcommand{\pln}{\ensuremath{\geom{R}^2}}
\newcommand{\sph}{\ensuremath{\geom{S}}}
Each set in $\fami{F}$ is called \emph{open}. A \emph{closed set} is the complement of an open set. (The empty set and $X$ are both open and closed.)

We consider %
the topological space $\geom{R}^\ell = (\univ{R}^\ell, \fami{F})$,
where $\fami{F}$ is the standard topology of $\mathbb{R}^\ell$, that is, $\fami{F}$ is the closure under union and finite intersection of the set containing the open ball~$\{\vec{x} \in \mathbb{R}^\ell \mid \lVert\vec{x} - \vec{y}\rVert < d\}$ for each $d \in \mathbb{R}$, $\vec{y} \in \mathbb{R}^\ell$,
where 
$\lVert\cdot\rVert$ is the Euclidean norm.

A \emph{topological subspace}~$\geom{Y} \subseteq \geom{X}$ of a topological space~$\geom{X}$ is a topological space whose universe is a subset of the universe of $\geom{X}$. We always assume topological subspaces to carry the \emph{subspace topology}, that is, the open sets of $\geom{Y}$ are the intersections of the open sets of $\geom{X}$ with the universe of $\geom{Y}$. We also say that $\geom{Y}$ is the topological subspace \emph{induced} by the universe of $\geom{Y}$.

Important topological subspaces of $\geom{R}^\ell$ are, with a slight abuse of notation,
\begin{compactitem}
\item the \emph{plane}~$\geom{R}^2$,
\item the \emph{sphere}, whose universe is $\{(x, y, z) \in \mathbb{R}^3 \mid x^2 + y^2 + z^2 = 1\}$,
\item the \emph{closed disk}, whose universe is $\{(x, y) \in \mathbb{R}^2 \mid
  x^2 + y^2 \leq 1\}$,
\item the \emph{open disk}, whose universe is $\{(x,
  y) \in \mathbb{R}^2 \mid x^2 + y^2 < 1\}$, and
\item the \emph{circle},
  whose universe is $\{(x, y) \in \mathbb{R}^2 \mid x^2 + y^2 = 1\}$.
\end{compactitem}
A \emph{homeomorphism}~$\phi$ between two topological spaces is a bijection $\phi$ between the two corresponding universes such that both $\phi$ and $\phi^{-1}$ are continuous. We often refer to a \emph{subspace $\geom{X}$ in a topological space~$\geom{Y}$} (for example, a circle on a sphere), by which we mean a topological subspace of $\geom{Y}$ which is homeomorphic to~$\geom{X}$.

An \emph{arc} is a topological space that is homeomorphic to the closed interval $[0,1] \subseteq \geom{R}^1$. The images of 0 and 1 under a corresponding homeomorphism are the \emph{endpoints} of the arc, which \emph{links} them and runs \emph{between} them. Let $\geom{X} = (X, \fami{F})$ %
be a topological space. Being linked by an arc in $\geom{X}$ defines an equivalence relation on~$X$. The topological subspaces induced by the equivalence classes of this relation are called \emph{regions}. We say that a closed set $C$ in a topological space~$\geom{S}$ \emph{separates} $\geom{S}$ into the regions of the subspace of $\geom{S}$ induced by $S \setminus C$, where $S$ is the universe of $\geom{S}$.

For more on topology, see \citet{Mun00}, for example.

\paragraph{Embeddings of graphs into the plane and sphere.}
An \emph{embedding} of a graph~$G = (V, E)$ into the plane~\pln\ (into the sphere~\sph) is a tuple $(\geom{V}, \fami{E})$ and a bijection $\phi \colon V \to \geom{V}$ such that
\begin{compactitem}
\item $\geom{V} \subseteq \pln$ ($\geom{V} \subseteq \sph$),
\item $\fami{E}$ is a set of arcs in \pln\ (in \sph) with endpoints in $\geom{V}$,
\item the interior of any arc in $\fami{E}$ (that is, the arc without its endpoints) contains no point in $\geom{V}$ and no point of any other arc in $\fami{E}$, and
\item $u, v \in V$ are adjacent in $G$ if and only if $\phi(u)$ is linked to $\phi(v)$ by an arc in~$\fami{E}$.
\end{compactitem}
The regions in $\pln \setminus (\bigcup \fami{E})$ (in $\sph \setminus (\bigcup \fami{E})$) are called \emph{faces}.

A \emph{planar graph} is a graph which has an embedding in the plane or, equivalently, in the sphere.
A \emph{minor} of a graph $G$ is a graph obtained from a subgraph of~$G$ by contracting edges, that is, replacing the two endpoints of an edge~$\{u, v\}$ by a new vertex which is adjacent to all vertices in~$N(u) \cup N(v) \setminus \{u, v\}$.
It follows from Kuratowski's theorem that a graph is planar if and only if it does not have a $K_5$ or a $K_{3, 3}$ as a minor~\cite[Section 4.4]{Die16}.
A~\emph{plane graph}~$G = (V, E)$ is a planar graph together with a fixed embedding in the plane.
An \emph{\splane\ graph}~$G$ is a planar graph given with a fixed embedding in the sphere. For notational convenience, we refer to the sets~$V$ and $\geom{V}$ as well as~$E$ and $\fami{E}$ interchangeably. Moreover, we sometimes identify $G$ with the set of points $\geom{V} \cup \bigcup \fami{E}$.

A \emph{noose in an \splane\ graph~$G$} is a circle in~$\geom{S}$ whose intersection with~$G$ is contained in $V(G)$. Note that every noose separates $\geom{S}$ into two open disks.

\paragraph{Layer decompositions, outerplanar graphs, face paths.}
The face of unbounded size in the embedding of a plane graph~$G$ is called \emph{outer face}. The \emph{layer decomposition} of~$G$ with respect to the embedding is a partition of~$V$ into layers~$L_1\uplus\dots\uplus L_r$ and is defined inductively as follows. Layer~$L_1$ is the set of vertices that lie on the outer face of~$G$.
For each $i \in \{2, \ldots, r\}$, layer~$L_i$ is the set of vertices that lie on the outer face of~$G - (\bigcup_{j=1}^{i-1} L_j)$. The graph $G$ is called \emph{$r$-out\-er\-pla\-nar} if it has an embedding with a layer decomposition consisting of at most $r$~layers. The \emph{outerplanarity number} of $G$ is the minimum~$r$ such that $G$~is \loutp. If $r=1$, then $G$~is said to be \emph{outerplanar}. A \emph{face path} is an alternating sequence of faces and vertices such that two consecutive elements are incident with one another. The first and last element of a face path are called its \emph{ends}. Note that the ends of a face path may be two vertices, two faces, or a face and a vertex. The \emph{length} of a face path is the number of faces in the sequence. Note that a vertex~$v$ in layer~$L_i$ has a face path of length~$i$ from~$v$ to the outer face. Moreover, a graph is \loutp\ if and only if each vertex has a face path of length at most~$r$ to the outer face.

\paragraph{Branch decompositions.}
A \emph{branch decomposition of a graph~$G$} is a pair~$(T, \lambda)$,
where $T$ is a ternary tree, that is, each internal vertex has degree three,
and $\lambda$ is a bijection between the leaves of~$T$ and~$E(G)$. Every edge $e \in E(T)$ defines a bipartition of $E(G)$ into $A_e, B_e$ corresponding to the leaves in the connected components of~$T - e$. Define the \emph{middle set~$M(e)$} of an edge~$e \in E(T)$ to be the set of vertices in~$G$ that are incident with both an edge in $A_e$ and $B_e$. That is, \[M(e) \mcoloneqq \{v \in V(G) \mid \exists a \in A_e \exists b \in B_e \colon v \in a \cap b\}.\]
The \emph{width of an edge $e \in E(T)$} is $|M(e)|$ and the \emph{width of a branch decomposition~$(T, \lambda)$} is the largest width of an edge in~$T$. The \emph{branchwidth of a graph~$G$} is the smallest width of a branch decomposition of~$G$.

A \emph{\scbd\ of an \splane\ graph~$G$} is a branch decomposition~$(T, \lambda)$ of~$G$ fulfilling the following additional condition. For each edge~$e \in E(T)$, there is a noose~$\geom{N}_e$ whose intersection with $G$ is precisely $M(e)$ and, furthermore, the open disks~$\geom{D}_1, \geom{D}_2$ into which the noose~$\geom{N}_e$ separates~$\geom{S}$ satisfy $\geom{D}_1 \cap G = A_e \setminus M(e)$ and $\geom{D}_2 \cap G = B_e \setminus M(e)$.
We use the following~theorem.
\begin{theorem}[\cite{ST94,DorPBF10,MP15}]\label{scbd}
  Let $G$ be a connected, bridgeless, \splane\ graph of branchwidth at most~$b$. There exists a \scbd\ for $G$ of width at most $b$.%
\end{theorem}
\noindent\citet{DorPBF10} first noted that \citet{ST94} implicitly proved a variant of \cref{scbd} in which $G$ is required to have no degree-one vertices rather than no bridges. \citet{MP15} observed a flaw in \citeauthor{DorPBF10}'s derivation, showing that bridgelessness is required (and sufficient). The \scbd\ in \cref{scbd} can be computed in $\Oh(|V(G)|^3)$~time (see \citet{GT08}), but we do not need to explicitly construct it.%

\paragraph{Parameterized algorithms.}
Let $\Sigma$ be an alphabet.
A \emph{parameter} is a mapping $\Sigma^* \to \univ{N}$.
For a string $q$, $\kappa(q)$ is the \emph{parameter value}.
A \emph{parameterized problem} is a tuple $(Q, \kappa)$ of a language~$Q$ over some alphabet~$\Sigma$ and a parameter $\kappa$.
We say that an algorithm is a \emph{fixed-parameter algorithm} with respect to a parameter $\kappa$ if the algorithm has running time $\phi(\kappa(q)) \cdot \poly(|q|)$ where~$q$ is the input and $\phi$~is some computable function.

To simplify notation, we omit explit reference to the function $\kappa$.
For example, if we have a problem of finding a solution of size~$k$, then we write ``$k$'' for the solution size parameter, that is, the mapping that takes an instance and extracts the value of~$k$.
Moreover, when specifying running times with respect to a parameter~$\kappa$, we often replace~$\kappa(q)$ by the referenced value if the instance~$q$ is clear from the context.

A core tool in the development of fixed-parameter algorithms is polynomial-time preprocessing by \emph{data reduction}~\cite{GN07SIGACT,Bod09,Kra14}.
The goal is to remove needless information from the input so to reduce its size or to obtain some desirable properties of the input.
Such small or well-formed instances can then be exploited by algorithms that produce a solution: small size of the input implies a small search space of the solution algorithm, and similarly, well-formed instances may be easier to solve.

Data reduction is usually presented as a series of \emph{reduction
  rules}.
These are poly\-nomial-time algorithms that take as input an
instance of some decision problem and produce another instance of the
same problem as output.
A reduction rule is \emph{correct} if for each
input instance~$I$, the corresponding output instance of the rule is a
\yesin\ if and only if~$I$ is a \yesin.
We call an
instance~$I$ of a parameterized problem \emph{reduced} with respect to
a reduction rule if the reduction rule does not apply to~$I$.
That is,
carrying out that reduction rule yields an unchanged instance.

The notion of problem kernels captures the idea of reduction rules with effectiveness guarantee.
A \emph{kernelization} or \emph{problem kernel} for a parameterized problem $(Q, \kappa)$ is a parameterized reduction~$\rho$ from $(Q, \kappa)$ to itself such that $\rho$ is computable in polynomial time and there is a function $\phi$ such that for every $q \in \Sigma^*$ we have $|\rho(q)| \leq \phi(\kappa(q))$.
We also call $\phi$ the \emph{size} of $\rho$.
If $\phi$ is polynomial, then we also call $\rho$ a \emph{polynomial kernelization} or \emph{polynomial problem~kernel}.

\section{Beware of removing twins}\label{sec:twins-planar}

\newcommand{\counterexcommon}{
  \node (h1) at (0,0) [vert,label={[yshift=1mm]right:$h_1$}] {};
  \node (b) at (0,-1) [vert,label=left:$b$] {};
  \node (y) at (0,-2) [vert,label=left:$y$] {};
  \node (h2) at (0,-3) [vert,label={[yshift=-1mm]right:$h_2$}] {};

  \node (x) at (-1,-1) [vert,label=left:$x$] {};
  \node (a) at (-1,-2) [vert,label=left:$a$] {};

  \node (z) at (1,-1) [vert,label=right:$z$] {};
  \node (c) at (1,-2) [vert,label=right:$c$] {};

  \draw [-,thick] (h1)-- (b) -- (y) -- (h2);

  \draw [-,thick] (h1) -- (x) -- (a) -- (h2);

  \draw [-,thick] (h1) -- (z) -- (c) -- (h2);

  \node (t1) at (-0.61, -1.5) [vert,label=above:$t$] {};
}

\begin{figure}[t]\centering
  \hfill
  \begin{tikzpicture}[>=stealth,x=55pt,y=35pt,
    ivert/.style={vert},%
    tvert/.style={vert}]%

    \counterexcommon{}
    \node (t2) at (0.61, -1.5) [vert,label=above:$t'$] {};

    \draw [dotted,thick] (t1) -- (a);
    \draw [dotted,thick] (t1) -- (x);

    \draw [dotted,thick] (t1) -- (b);
    \draw [dotted,thick] (t1) -- (y);

    \draw [dotted,thick] (t2) -- (b);
    \draw [dotted,thick] (t2) -- (y);

    \draw [dotted,thick] (t2) -- (c);
    \draw [dotted,thick] (t2) -- (z);
  \end{tikzpicture}%
  \hfill%
  \begin{tikzpicture}[>=stealth,x=55pt,y=35pt,
    ivert/.style={vert},%
    tvert/.style={vert}]%

    \tikzstyle{heda} = [-, fill = green!80, fill opacity = 0.3, thick]
    \tikzstyle{hedb} = [dashed, fill = magenta!80, fill opacity = 0.3, thick]

    \input{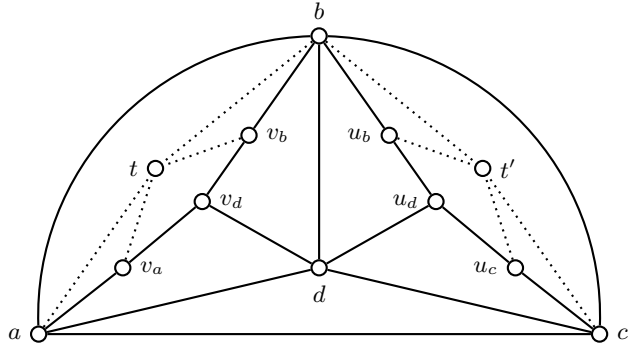}

    \counterexcommon{}

    \draw [-,thick] (b) -- (c);

    \draw [-,thick] (y) -- (z);

    \begin{pgfonlayer}{background}
      \draw [-,thick, heda] \hedgei{z}{2mm};
      \draw [-,thick, heda] \hedgei{b}{2mm};
      \draw [-,thick, heda] \hedgeii{a}{h2}{2mm};

      \draw [-,thick, hedb] \hedgei{y}{2mm};
      \draw [-,thick, hedb] \hedgei{c}{2mm};
      \draw [-,thick, hedb] \hedgeii{x}{h1}{2mm};
    \end{pgfonlayer}
  \end{tikzpicture}%
  \hfill\mbox{}%
  \caption{Left: A hypergraph~$\Hyp$ and its support, showing that twins can be essential for obtaining a $2$-outerplanar support. The set of hyperedges consists of size-two hyperedges that are drawn as solid lines between the corresponding vertices and, additionally, $\{a,b,t,t'\}$, $\{b,c,t,t'\}$, $\{x,y,t,t'\}$, and $\{y,z,t,t'\}$. Note that the vertices~$t$ and $t'$ are twins. The hypergraph~$\Hyp$~has a (2-outer)planar support whose edges are indicated by the solid and the dotted lines. However, $\Hyp-t'$ does not have a planar support, as witnessed on the right: A $K_{3, 3}$-minor of a support of $\Hyp - t'$ obtained in the proof of \cref{thm:counterexample}. One partite vertex set of the minor is encircled with dashed lines and colored in red, the other partite vertex set is encircled solidly and colored in green.
  }
\label[figure]{fig:counter}
\end{figure}

In \cref{fig:counter}, we provide a concrete example that shows that twins can be necessary to obtain a 2-outer-planar support. More precisely, the hypergraph on the left in \cref{fig:counter} witnesses the following theorem. Recall that, for a hypergraph $\Hyp$ and a vertex $v$ in $\Hyp$, the hypergraph $\Hyp - v$ is obtained from $\Hyp$ by removing $v$ out of each hyperedge in $\Hyp$.

\begin{theorem}\label{thm:counterexample}
  There is a hypergraph $\Hyp$ with two twins~$t$ and $t'$ such that $\Hyp$ has a 2-outerplanar support, but the hypergraph~$\Hyp - t'$ does not have a planar support. Each hyperedge in $\Hyp$ has size at most four.
\end{theorem}
\noindent The hypergraph improves on the example given in the conference version of this article~\cite{BKK+17} by having fewer vertices, hyperedges, and smaller maximum hyperedge size.
\begin{proof}[Proof of \cref{thm:counterexample}]
  The hypergraph~$\Hyp$ has vertex set $\{a, b, c, x, y, z, t, t', h_1, h_2\}$.
  The hyperedge set consists of the size-two hyperedges shown as solid lines on the left in \cref{fig:counter} and, additionally, the four hyperedges $\{a,b,t,t'\}$, $\{b,c,t,t'\}$, $\{x,y,t,t'\}$, and $\{y,z,t,t'\}$.

  Consider the graph~$G$ shown on the left of \cref{fig:counter} that includes both the dashed and solid edges.
  Observe that each of the size-four hyperedges of $\Hyp$ induces a path in $G$ and thus $G$ is a support of~$\Hyp$.
  Moreover, $G$ is 2-outerplanar.

  It remains to prove that $\Hyp - t'$ does not have a planar support.
  Suppose, for the sake of contradiction, that there is a planar support~$G$.
  Due to the hyperedges of size two in $\Hyp - t'$, there are two cycles in $G$, namely $C_1$, defined as $(x, h_1, b, y, h_2, a)$ and $C_2$, defined as $(z, h_1, b, y, h_2, c)$.
  Consider a planar embedding of~$G$.
  Consider the circle subspaces $\geom{C}_1$ and $\geom{C}_2$ in the plane defined by the embedding of $C_1$ and $C_2$, respectively.
  Circles $\geom{C}_1$ and $\geom{C}_2$ separate the plane into three open regions.
  These three open regions define three smallest enclosing closed regions:
  \begin{compactitem}
  \item The region~$\geom{R}_1$, containing the vertices of $C_1$ and no vertices in $V(C_2) \setminus V(C_1)$, 
  \item region~$\geom{R}_2$, containing the vertices of $C_2$ and no vertices in $V(C_1) \setminus V(C_2)$, and
  \item region~$\geom{R}_3$, containing $h_1$, $h_2$, the vertices in the symmetric difference $(V(C_1) \setminus V(C_2)) \cup (V(C_2) \setminus V(C_1))$ of $C_1$ and $C_2$, and no vertex in $(V(C_1) \cap V(C_2)) \setminus \{h_1, h_2\}$.
  \end{compactitem}
  Vertex~$t$ is embedded in one of these three regions.
  Because of symmetry it suffices to obtain a contradiction in the case where $t$ is not embedded in region~$\geom{R}_2$, that is, the case where $t$ is not embedded in region~$\geom{R}_1$ is analogous.
  Assume thus that $t$ is not embedded in region~$\geom{R}_2$.
  Then, $t$ is adjacent in $G$ with at most one of $\{b, c\}$ and with at most one of $\{y, z\}$, by planarity of $G$ and the definition of the three regions.
  Moreover, $G[\{b, c, t\}]$ and $G[\{y, z, t\}]$ are connected, because~$\{b, c, t\}$ and $\{y, z, t\}$ are hyperedges of $\Hyp - t'$, and thus 
  $G$ contains the edges $\{b, c\}$ and $\{y, z\}$.
  Thus, $G$ contains the graph shown on the right in \cref{fig:counter} as a subgraph.
  By contracting $\{x, h_1\}$ and $\{a, h_2\}$ in $G$ we obtain a $K_{3, 3}$ as a minor, contradicting the fact that~$G$ is planar.
  Thus, indeed, $\Hyp - t'$ does not have a planar support.
\end{proof}

\noindent
\cref{thm:counterexample} shows that
removing
one vertex of a twin class can transform a \yesin\ of \pPS\ into a
\noin.
We next address two features of the construction used in \cref{thm:counterexample}.
First, we show that the counterexample requires hyperedges of size at least four.
Second, we show that counterexamples exist with twin classes of arbitrary size, rather than just with pairs of twins.

\paragraph{Hyperedge size at most three.}
We now show that removing twins is correct, as long as they are
contained only in hyperedges of size at most three. That is, the following data reduction rule for \PS\ is correct.
\begin{rrule}\label{rr:size-three}
  Let $\Hyp$ be a hypergraph and $u$ and $v$ be twins in $\Hyp$ such that each hyperedge containing $u$ and $v$ has size at most three.
  Then, remove $u$ from~$\Hyp$.
\end{rrule}

\begin{theorem}
  \cref{rr:size-three} is correct, that is, $\Hyp$ has an \loutp\ support if and only if $\Hyp - u$ has.%
\end{theorem}
\begin{proof}
  If $\Hyp - u$ has an $r$-outerplanar support~$G$, then $\Hyp$ has an $r$-outerplanar support~$G'$: add $u$ as a degree-one neighbor to $v$ in $G$ to obtain~$G'$.
  It remains to show that, if $\Hyp$ has an $r$-outerplanar support~$G$, then $\Hyp - u$ has an $r$-outerplanar support.

  First,
  assume that
  $\{u, v\} \in E(G)$.
  Then $\Hyp - u$ has an $r$-outerplanar support~$G'$:
  Obtain $G'$ from $G$ by contracting the edge $\{u, v\}$ onto $v$.
  Since contracting an edge cannot create new faces in the corresponding embedding,
  $G'$ is $r$-outerplanar.
  Clearly, each hyperedge induces a connected graph in~$G'$.

  Now assume that $\{u, v\}\notin E(G)$.
  Consider an arbitrary hyperedge $F$ containing $u$ and $v$ in~$\Hyp$ and observe that $F$ has size three.
  Let $w \in F \setminus \{u, v\}$.
  Since $u$ and $v$ are nonadjacent in~$G$ yet~$G[F]$~is connected, we have $\{u, w\}, \{v, w\} \in E(G)$.
  Thus, $G - u$ is an $r$-outerplanar support for~$\Hyp - u$, as required.
\end{proof}

\paragraph{Large twin classes.} To show that the example from
\cref{thm:counterexample} is not a pathology of having only one pair of twins, we now extend it so that an arbitrarily large set of twins is required for the existence of a planar support.
The basic observation is that in \cref{thm:counterexample} we indeed have proved that there are two disjoint regions in which the two twins have to reside, for otherwise there is no planar support.
By introducing several copies of the hypergraph constructed in \cref{thm:counterexample} and merging them in an appropriate way, we can ensure that there is an arbitrarily large number~$\ell$ of regions, of which each requires its own private vertex from a twin class of size~$\ell$.

\begin{theorem}\label{thm:manytwins}
  For each integer $\ell \geq 2$, there is a hypergraph $\Hyp$ with a set~$X$ of $2\ell$ mutual twins such that $\Hyp$ has a 2-outerplanar support, but for each $t \in X$ the hypergraph~$\Hyp - t$ does not have a planar support.
\end{theorem}

\begin{proof}
  Fix an integer $\ell \in \univ{N}$.
  To construct the hypergraph $\Hyp$, copy~$\ell$ times the vertex set~$V(\Hyp)$ from the proof of \cref{thm:counterexample}, and let~\(V_i \coloneqq \{a_i, b_i, c_i, x_i, y_i, z_i, h^1_i, h^2_i, t_i, t_i'\}\) denote the vertex set of the~$i$th copy, $i \in \ivals{\ell}$.
  The hyperedges are defined as follows.
  Within each copy, add the size-two solid hyperedges as shown on the left of \cref{fig:counter}.
  Then, add a distinguished new vertex~$v^*$, and add the size-two hyperedges~$\{x_i, v^*\}$, $\{a_i, v^*\}$, $\{z_i, v^*\}$, and $\{c_i, v^*\}$.
  Intuitively,
  vertex~$v^*$ serves to force each of the copies
  into an embedding with the same outer face
  and as a conduit to connect subgraphs induced by hyperedges that contain vertices from each of the copies, which we are about to introduce.

  Denote by $A, B, C, X, Y, Z$ the sets of copies of the corresponding vertices above, that is, $A:=\{a_i\mid i \in \ivals{\ell}\}$, $B:=\{b_i\mid i \in \ivals{\ell}\}$, and so forth. %
  Let $T:=\{t_i, t'_i \mid i \in \ivals{\ell}\}$ denote the set of designated twins.
  The final hyperedges in $\Hyp$ are
  \begin{align*}
    A\cup B \cup T \cup \{v^*\}, && B \cup C \cup T \cup \{v^*\},\\
    X \cup Y \cup T \cup \{v^*\}, & & \text{and } Y\cup Z \cup T \cup \{v^*\}.
  \end{align*}
  Note that $T$ forms a twin class in the resulting
  hypergraph~$\Hyp$. Hypergraph~$\Hyp$ has a \outp{2}\
  support
  because~$v^*$ can be used to join in a star-like fashion into one connected graph the partial
  supports for each $\Hyp[V_i]$ that are obtained by copying the support for the example on the left of \cref{fig:counter}.

  \looseness=-1
  We claim that, for each $t \in T$, hypergraph~$\Hyp - t$ does not have a planar support.
  Assume, for the sake of a contradiction, that there is a planar support~$G$.
  Choose a planar embedding of $G$ such that $v^*$ is incident with the outer face.
  Due to the hyperedges of size two in $\Hyp - t'$, there are $2\ell$ cycles in $G$, namely, for each $i \in [\ell]$, there is the cycle $C^1_i$, defined as $(x_i, h^1_i, b_i, y_i, h^2_i, a_i)$ and the cycle $C^2_i$, defined as $(z_i, h^1_i, b_i, y_i, h^2_i, c_i)$.
  The embeddings of these cycles define $2\ell + 1$ closed regions as follows.
  The embedding of the cycle $C^\alpha_i$, for $\alpha \in \{1, 2\}$ and $i \in [\ell]$, separates the plane into two open regions.
  Of these two regions, pick the one that has empty intersection with the outer face, take the smallest enclosing closed region and denote it by $\geom{R}_{2i - 2 + \alpha}$.
  This defines $2\ell$ regions $\geom{R}_{1}, \ldots, \geom{R}_{2\ell}$.
  Region $\geom{R}_0$ is defined by taking the plane, removing from it all points in the regions $\geom{R}_{1}, \ldots, \geom{R}_{2\ell}$, and then taking the smallest enclosing closed region.

  Since $v^*$ is incident with the outer face, and because of the edges $\{x_i, v^*\}$, $\{a_i, v^*\}$, $\{z_i, v^*\}$, and~$\{c_i, v^*\}$, vertex~$v^*$ is not contained in any region $\geom{R}_i$ such that $i \in [2\ell]$.
  It follows that the regions~$\geom{R}_i$ are pairwise disjoint except for their boundary cycles.
  Thus, the regions $\geom{R}_i$, $i \in \{0, \ldots, 2\ell\}$, satisfy the following properties:
  \begin{compactitem}
  \item The region $\geom{R}_0$ contains all vertices of each cycle $C^\alpha_i$, $\alpha \in \{1, 2\}$, $i \in [\ell]$, except for the vertices in $B \cup Y$.
  \item For each $i \in [\ell]$, region $\geom{R}_{2i - 1}$ contains all vertices of $C^1_i$, no vertex in $V(C^2_i) \setminus V(C^1_i)$, and no vertex of any other cycle $C^1_j$, $C^2_j$ for $j \in [2\ell] \setminus \{i\}$.
  \item For each $i \in [\ell]$, region $\geom{R}_{2i}$ contains all vertices of $C^2_i$, no vertex in $V(C^1_i) \setminus V(C^2_i)$, and no vertex of any other cycle $C^1_j$, $C^2_j$ for $j \in [2\ell] \setminus \{i\}$.
  \end{compactitem}
  Since there are only $2\ell - 1$ twins from $T$ present in $\Hyp - t$, one of the regions $R_i$, $i \in \{1, \ldots, 2\ell\}$ does not contain a vertex of~$T$, say $R_{i^\circ}$.
  By symmetry we may assume that $i^\circ$ is even.

  Consider the hyperedge $F \coloneqq B \cup C \cup (T \setminus \{t\}) \cup \{v^\star\}$
  and the two vertices $b_{i^\circ}$ and $c_{i^\circ}$ in $F$.
  Observe that the only region $\geom{R}_i$, $i \in \{0, \ldots, 2\ell\}$, containing $b_{i^\circ}$ and another vertex of $F$ is $\geom{R}_{i^\circ}$, and the only other vertex of $F$ in $\geom{R}_{i^\circ}$ is $c_{i^\circ}$.
  Thus, since $G[F]$ is connected, there is an edge $\{b_{i^\circ}, c_{i^\circ}\}$ in $G$.
  By an analogous argument for the hyperedge $Y \cup Z \cup (T \setminus \{t\}) \cup \{v^{\star}\}$
  we obtain that $\{y_{i^\circ}, z_{i^\circ}\} \in E(G)$.
  Thus, $G[V_{i^\circ}]$ contains a subgraph as shown on the right of \cref{fig:counter}.
  Thus, $G$ contains a $K_{3, 3}$-minor contradicting the fact that $G$ is planar.
  Hence, $\Hyp - t$ does not have a planar support, as claimed.
\end{proof}

\section{The existence of an upper bound on the number of important twins}\label{sec:nonuni-fpt}

In the proof of \cref{thm:manytwins},
the number of hyperedges increases with the number of necessary twins we seek to enforce. We now show that this is unavoidable. That is, with a fixed number of hyperedges, it is impossible to create arbitrarily large twin classes out of which no twin can be deleted without violating the property of having a planar support.

\begin{theorem}\label{nonuni-fpt}
  There exists a function $\psi \colon \mathbb{N} \to \mathbb{N}$ with
  the following property. For each $m \in \mathbb{N}$ and every
  hypergraph~$\Hyp$ that has at most $m$~hyperedges, out of each twin class
  of~$\Hyp$, we can remove all but $\psi(m)$ arbitrary twins such that
  the resulting hypergraph has a planar support if and only if $\Hyp$
  has one.
\end{theorem}

The basic observation is that adding twins is not detrimental. If we have a planar support for $\Hyp$, then we can make a new twin adjacent to one of its already present twins, so that the resulting graph remains planar. Reversing this idea, from each hypergraph with a planar support, by deleting twins we can obtain  a \emph{minimal} hypergraph~$\hyp'$ which also has a planar support but from which no further twins can be deleted while maintaining the property of having a planar support. Using Dickson's lemma (see below for details) it is not hard to show that there is a function $\phi$ such that, for each fixed number~$m$ of hyperedges, there are only $\phi(m)$ such minimal hypergraphs. %
Clearly, among these minimal hypergraphs, one has a largest twin class, whose size we can put as the value of~$\psi(|\HE(\Hyp)|)$.

We now formalize the above approach. Denote by $\univ{S}$ the set of hypergraphs which have a planar support. As mentioned, $\univ{S}$ is \emph{closed under adding twins}, that is, taking an arbitrary hypergraph in~$\univ{S}$ and adding a twin to it yields another hypergraph in~$\univ{S}$.

\begin{proof}[Proof of \cref{nonuni-fpt}]
  We first define a quasi-order~$\preceq$ on the family of hypergraphs with at most $\nEd{}$~hyperedges. (A \emph{quasi-order} is reflexive and transitive.) To define~$\preceq$, we say that $\Hyp \preceq \G$ if $\Hyp$~can be obtained from~$\G$ by iteratively removing a vertex that has a twin. If we allow zero removals so that $\preceq$~is reflexive, it is clear that $\preceq$~is a quasi-order. Moreover, if $\Hyp \in \univ{S}$ and $\Hyp \preceq \G$, then $\G \in \univ{S}$ since $\univ{S}$ is closed under adding twins.

  For every~$m \in \mathbb{N}$ let~$\univ{F}_m$ denote the family of hypergraphs in $\univ{S}$ that contain at most $\nEd{}$ hyperedges and are minimal under~$\preceq$.
  Next we show that $\univ{F}_m$ is finite.
  Consider the representation of a hypergraph~$\Hyp$ with at most $\nEd{}$ hyperedges as a $2^{\nEd{}}$-tuple~$t_\Hyp \in \univ{N}^{2^\nEd{}}$, each entry of which represents the size of a distinct twin class. The set of such tuples is quasi-ordered by a natural extension of~$\leq$, namely $(a_1, \ldots, a_\ell) \leq (b_1, \ldots, b_\ell)$ if and only if $a_i \leq b_i$ for each $i \in \{1, \ldots, \ell\}$. We now lead the assumption that $\univ{F}_m$ is infinite to a contradiction by using Dickson's lemma~\cite[Lemma A]{Dic13}.

  If $\univ{F}_m$ is infinite, then there is an infinite subset $\univ{F}_m'$ of hypergraphs which have the same (nonempty) twin classes. That is, the tuples representing the hypergraphs of~$\univ{F}_m'$ have the same $0$-entries. For hypergraphs~$\Hyp, \G$ with the same twin classes, $t_\Hyp \leq t_\G$ implies $\Hyp \preceq \G$. Thus, $\univ{F}_m'$ gives an infinite set~$T$ of tuples that are pairwise incomparable under $\leq$. Dickson's lemma states that for every set $S \subseteq \univ{N}^{\ell}$ there exists a finite subset $S' \subseteq S$ such that for each $s \in S$ there is an $s' \in S'$ with $s' \leq s$. This is a contradiction to $T$ containing infinitely many incomparable tuples. Hence, $\univ{F}_m$ is finite.

  Finally, we simply choose the function~$\psi(m)$ in \cref{nonuni-fpt}
 as the largest size of a twin class of a hypergraph in~$\univ{F}_m$.
\end{proof}

\noindent
Let us briefly consider the implications of \cref{nonuni-fpt} for our ultimate goal---designing algorithms that compute for any given hypergraph a planar support if there is one and have fixed-parameter running time with respect to the number~$\nEd$ of hyperedges in the input hypergraph.
\Cref{nonuni-fpt} does imply that, for each $\nEd \in \univ{N}$, there is such an algorithm which works for all input hypergraph with at most $\nEd$ edges: 
The algorithm has the minimal yes-instances (in $\univ{F}_m$ from the proof) as hard-coded constants and checks whether the given hypergraph $\G$ satisfies $\Hyp \preceq \G$ for some $\Hyp \in \univ{F}_\nEd$.
However, a priori we do not not know how to compute the set~$\univ{F}_m$ of minimal \yesin s, making this result nonconstructive and thus not useful in practice. Furthermore, we would like to instead have one algorithm which works for any input hypergraph.
Hence, to eventually obtain implementable algorithms that are able to deal with any input, it is important to constructivize \cref{nonuni-fpt}. Below, we provide such a constructivization for \loutp\ supports.

\section{Relevant twins for $\lrs$-outerplanar supports}\label{sec:uniform-fpt}\label{sec:application}
Towards making the result of \cref{nonuni-fpt}
algorithmically exploitable,
we now give an explicit upper bound on the function~$\psi$.
Concretely, we prove that
out of each twin class of a hypergraph~$\Hyp$, we can
remove all but $\psi(m, r)$ twins such that the resulting hypergraph
has an $r$-outerplanar support if and only if~$\Hyp$ has. In other
words, we prove that the following data reduction rule is correct.
\begin{rrule}\label{ps:rr}
  Let $\Hyp$~be a hypergraph with $m$~edges. If there is a twin class with more than $\psi(m, r) = \suppsize$ vertices, then remove one vertex out of this class.
\end{rrule}
\noindent
Assuming that \cref{ps:rr} is correct, \cref{thm:lintime} follows.
\begin{proof}[Proof of \cref{thm:lintime}]
  \cref{ps:rr} can be applied exhaustively in linear time because the
  twin classes can be computed in linear time~\cite{HPV99}. After
  this, each twin class contains at most $\psi(m, r)$ vertices,
  meaning that, overall, at most $2^m\psi(m, r)$
  vertices remain. Testing all possible planar graphs for whether they are a
  support for the resulting hypergraph thus takes constant time if $m$
  and $r$~are constant. Hence, the overall running time is linear in
  the input size.
\end{proof}
\noindent
We mention in passing that, in the terms of parameterized
algorithmics, exhaustive application of \cref{ps:rr} yields a so-called
problem kernel.

The correctness proof for \cref{ps:rr} consists of two parts. First, in \cref{sepseq}, we show that each \loutp\ graph has a long sequence of nested separators. Here, \emph{nested} means that each separator separates the graph into a \emph{left} side and a \emph{right} side, and each left side contains all previous left sides. Furthermore, the sequence of separators has the additional property that, for any pair of separators~$S_1, S_2$, we can glue the left side of $S_1$ and the right side of $S_2$, obtaining another \loutp\ graph.

In the second part of the proof, we fix an initial \loutp\ support for our
input hypergraph. We then show that, in a long sequence of nested
separators for this support, there are two separators such
that we can carry out the following procedure. We discard all vertices
between the separators, glue their left and right sides, and reattach
the vertices which we discarded as degree-one vertices. Furthermore,
we can do this in such a way that the resulting graph is an \loutp\
support. The reattached degree-one vertices hence are not crucial to
obtain an \loutp\ support. We will show that if our input hypergraph is
larger than \(\psi(m,r)\),
then there is always at least one vertex which can be discarded because it is between two suitable separators.

We now formalize our approach. \cref{sepseq} will guarantee the
existence of a long sequence of gluable separators; it is proven in
\cref{sec:sepseq}. To formally state it, we need the following notation.

\begin{definition}[Middle set, subgraph induced by an edge set]
For an edge bipartition~$A \uplus B = E(G)$ of a graph $G$, let $M(A, B)$ be the set of vertices in~$G$ which are incident with both an edge in $A$ and in $B$, that is, \[M(A,B) \mcoloneqq \{v \in V(G) \mid \exists a \in A \exists b \in B \colon v \in a \cap b\}.\] We call $M(A, B)$ the \emph{middle set} of $A, B$. For an edge set $A \subseteq E(G)$, let $G\langle A \rangle \coloneqq (\bigcup A, A)$ be the subgraph induced by~$A$.
\end{definition}
\noindent
Observe that the definition of middle set for an edge partition corresponds exactly to the definition of middle sets in branch decompositions. %
Recall from \cref{sec:prelim} the definitions of graph gluing, boundary, and boundary labeling.

\newcommand{\sepseqthmstatement}{%
  For every connected, bridgeless, \loutp\ graph~$G$ with $n$~vertices, there is a sequence $((A_i, B_i, \beta_i))_{i = 1}^s$ where each pair $A_i, B_i \subseteq E(G)$ is an edge bipartition of~$G$ and
  $\beta_i$ is a bijection $M(A_i, B_i) \to \ivals{|M(A_i, B_i)|}$
  such that $s \geq \log(n)/\signumhalved$, and, for every $i, j$, $1 \leq i < j \leq s$,
  \begin{lemenum}
  \item $|M(A_i, B_i)| = |M(A_j, B_j)| \leq 2r$,
  \item $A_i \subsetneq A_j$, $B_i \supsetneq B_j$, and\label[stat]{sepseq:nested}
  \item $G\langle A_i \rangle \glue G\langle
    B_j \rangle$ is \loutp, where $G\langle A_i \rangle$ is
    understood to be \bdrd{\beta_i}\ and $G \langle B_j
    \rangle$ is understood to be \bdrd{\beta_j}.\label[stat]{sepseq:glueable}
  \end{lemenum}%
}
\begin{theorem}\label{sepseq}
  \sepseqthmstatement
\end{theorem}
To gain some intuition for \cref{sepseq} note that each $M(A_i, B_i)$ is a separator, separating its left side $G\langle A_i \rangle$ from its right side $G \langle B_i \rangle$ in $G$. \Cref{sepseq:nested} ensures that each left side contains all previous left sides, that is, the separators are nested. \Cref{sepseq:glueable} ensures that for any two separators in the sequence, we can glue their left and right sides and again obtain an \loutp\ graph. In this new graph, the vertices in between the separators are missing---these will be the vertices which are not crucial to obtain an \loutp\ support.

The reason why we can prove the lower bound on the length of the sequence is basically that \loutp\ graphs have a tree-like structure, whence large \loutp\ graphs have a long ``path'' in this structure, and a long path in such a structure induces many nested separators. From such a path we can pick the separators that are amenable to \cref{sepseq:glueable}.

We next formalize the crucial vertices for obtaining an \loutp\ support. These are the vertices in a smallest representative support, defined as follows.
Recall that a~vertex~$v$ in a hypergraph \emph{covers} a vertex~$u$ if $\Ed(u) \subseteq \Ed(v)$, where $\Ed(u)$ and $\Ed(v)$ are the hyperedges containing the corresponding vertices.
\begin{definition}[Representative support]\label[definition]{def:representative}
  Let $\Hyp$ be a hypergraph. A graph~$G$ is a \emph{representative
    support} for $\Hyp$ if $V(G) \subseteq V(\Hyp)$, graph $G$ is a
  support for subhypergraph~$\Hyp|_{V(G)}$ shrunken to $V(G)$, and
  every vertex in $V(\Hyp) \setminus V(G)$ is covered in $\Hyp$ by
  some vertex in~$V(G)$.
\end{definition}

\noindent
Using the sequence of separators from~\cref{sepseq}, we show that the size
of a smallest representative \loutp\ support is upper-bounded by a
function of $m$
and~$r$%
. To this end, we take an initial support, find two separators such that each
vertex in between can be removed and reattached as a neighbor of some vertex that covers it and is not in between the separators. This implies that the removed vertices need not be contained in a representative support. %
Intuitively, the
two separators have to have the same ``status'' with respect to the
hyperedges that cross them. We formalize this as follows (some further intuition is given after the definition).
\begin{definition}[Edge-bipartition signature]\label[definition]{def:sepsig}
  Let~$\Hyp = (\Ver,\Ed)$ be a hypergraph and let~$G$ be a representative planar support for $\Hyp$.  Let~$(A,B,\beta)$ be a triple where $(A, B)$ is an edge bipartition of $G$ and
  $\beta$ is a bijection $M(A, B) \to \ivals{|M(A,B)|}$. Let $\ell \coloneqq |M(A, B)|$. The \emph{signature} of $(A, B, \beta)$ is a triple $(\T, \phi, \C)$, where
  \begin{compactitem}
  \item $\T \coloneqq \{ \{F \in \Ed \mid u \in F\} \mid u\in \bigcup A\}$ is the family of sets of hyperedges that contain some vertex in~$\bigcup A$,
  \item $\phi: \ivals{\ell}\to \{[u]_\twnrl\mid u\in \Ver \}\colon j\mapsto[\beta^{-1}(j)]_\twnrl$
    maps each index of a vertex in~$M(A, B)$ to
    the twin class of that vertex, and
  \item
    $\C \coloneqq \{(F, \gamma_F) \mid F \in \Ed\}$, where $\gamma_F$ is the relation on $\ivals{\ell}$ defined by $(i, j) \in \gamma_F$ whenever both $\beta^{-1}(i), \beta^{-1}(j) \in F$ and $\beta^{-1}(i)$ is connected to $\beta^{-1}(j)$ in $G\langle B \rangle[F \cap \bigcup B]$.
    \end{compactitem}
  \end{definition}
  \noindent For easier notation, we will in the following refer to the sets in the family $\T$ in a edge bipartition signature simply as \emph{twin class}.

  Intuitively, if we have two separators~$M(A_i, B_i)$, $M(A_j, B_j)$, $i < j$, from \cref{sepseq}, then the left side, $G\langle A_i \rangle$, of~$M(A_i, B_i)$ and the right side, $G\langle B_j \rangle$, of $M(A_j, B_j)$ can be glued to obtain an \loutp\ graph.
  If the two separators additionally have the same edge-bipartition signature, then the three elements of the signature ensure in the following way that we maintain that the glued graph remains a representative support for the hypergraph~$\Hyp$.
  First, since there are the same twin classes on the two left sides $G\langle A_i \rangle$, $G\langle A_j \rangle$ due to element $\T$, and the left side $G\langle A_j \rangle$ of~$M(A_j, B_j)$ contains all vertices between the two separators, all removed vertices between the two separators are still covered in the new graph.
  Second, elements $\phi$ and $\C$ ensure that the connectivity relation of each hyperedge in $\bigcup B_i$ is the same as in $\bigcup B_j$.
  That is, regardless of how the concrete connectivity in the ``forgotten part'' between the two separators looks like, replacing $B_i$ with $B_j$ preserves the fact that each hyperedge induces a connected subgraph.
  
  We have the following upper bound on the number of different separator states.

\begin{lemma}
  In a sequence $((A_i, B_i, \beta_i))_{i = 1}^s$ as in \cref{sepseq} the number of distinct edge-bipartition signatures is upper-bounded by $\sepnum$.\label{obs:sig-num}
\end{lemma}

\begin{proof}
  Denote the signature of $(A_i, B_i, \beta_i)$ by $(\T_i, \phi_i, \C_i)$.
  There
  are at most~$2^{\nEd{}} - 1$~twin classes in~$\T_i$. Furthermore,
  for every $i, j$, $i<j$, we have~$A_i\subsetneq A_j$, which implies
  $\T_i\subseteq \T_j$. Thus, either $\T_i = \T_{i+1}$ or
  $\T_{i+1}$~comprises at least one additional twin class. Since the
  number of twin classes can increase at most~$2^{\nEd{}} - 1$ times,
  the number of different~$\T_i$ is less than~$2^{\nEd{}}$. Next,
  there are at most~$2^{\nEd{}}$ choices for a twin class for
  each~$\beta^{-1}(i) \in M(A_i, B_i)$, leading to at most
  $2^{{\nEd{}}\ell}$ different possibilities where~$\ell = |M(A_i, B_i)|$. For the last part of the
  signature, $\C_i$, for each $\gamma_e$ there are $2^{(\ell^2 -
    \ell)/2}$ possibilities, leading to $2^{m(\ell^2 - \ell)/2}$
  possibilities for~$\C_i$. Since the size~$\ell$ of the middle sets in
  \cref{sepseq} is at most~$2\layers$, the number of possible signatures is at~most~\(
    2^m \cdot 2^{2m\layers} \cdot 2^{m\cdot(2\layers^2 - \layers)}
    \meq \sepnum.
    \)
\end{proof}

\noindent
As before, let $\psi(m, r) \coloneqq \suppsize$.
\newcommand{\sss}{\ensuremath{\psi(m, r)}}
\begin{lemma}\label{lem:supp-size}
  If a hypergraph $\Hyp = (\Ver,\Ed)$ has an \loutp\ support, then it has a
  representative \loutp\ support with at most $\sss$ vertices.
\end{lemma}
\begin{proof}
  Let~$G=(W,E)$ be a representative \loutp\ support for~$\Hyp$ with
  the minimum number of vertices and fix a corresponding planar embedding. Assume towards a contradiction
  that~$|W|>\sss$. We show that there is a representative support
  for~$\Hyp$ with less than $\sss$ vertices.

  We aim to apply \cref{sepseq} to $G$. For this we need that $G$ is connected and does not contain any bridges. Indeed, if $G$ is not connected, then pick a vertex incident with the outer face for each connected component and add edges to make the picked vertices induce a path.
  This does not affect the outerplanarity number of~$G$ (although it adds bridges). If $G$ has a bridge~$\{u, v\}$, then proceed as follows. At least one of the ends of the bridge, say $v$, has degree at least two because~$|W| > \sss \geq 2$. One neighbor~$w \neq u$ of $v$ is incident with the same face
  as~$u$, because~$\{u, v\}$ is a bridge. Add the edge~$\{u, w\}$. Thus, edge~$\{u, v\}$ ceases to be a bridge. We can embed $\{u, w\}$ in such a way that the face~$\geom{F}$ incident with $u, v$, and $w$ is split into one face that is incident with only $\{u, v, w\}$ and devoid of any other vertex, and one face~$\geom{F}'$ that is incident with all the vertices that are incident with $\geom{F}$ including~$u, v$, and~$w$. This implies that each vertex retains its layer~$L_i$, meaning that $G$ remains \loutp. Thus, we may assume that $G$ is connected, bridgeless, and \loutp.

  Since~$G$ contains more than~$\sss$ vertices, there is a sequence~$\S = ((A_i, B_i, \beta_i))_{i = 1}^s$ as in~\cref{sepseq} of length at least
  \begin{equation*}
    s \mgeq \frac{\log(\sss)}{\signumhalved} \meq \frac{6r \cdot 2^{\nEd{}\cdot (2\lrs^2+\lrs+1)}\cdot \signumhalved}{\signumhalved} \meq 6r \cdot \sepnum.
  \end{equation*}
  Since there are less than \sepnum\ different signatures in \S\ (\cref{obs:sig-num}), there are $6r$ elements of \S\ with the same signature.
  Note that each middle set~$M(A_i, B_i)$ induces a planar graph in $G$ and, since $|M(A_i, B_i)| \leq 2r$, this graph has at most \(\max\{1, 3 |M(A_i, B_i)| - 6\} \leq \max\{1, 6r - 6\} \) edges.
  Recall that $A_i \subsetneq A_{i + 1}$ for each $i \in [s - 1]$, and thus, $|A_{j}| - |A_{i}| \geq j - i$ for each $i, j \in [s]$ with $i < j$.
  Thus, if $i, j \in [s]$ with $i + 6r \leq j$, $M(A_i, B_i)$ and $M(A_j, B_j)$ differ in at least one vertex.
  Thus, there are two edge bipartitions $(A_i,B_i,\beta_i)$ and~$(A_j,B_j,\beta_j)$, $i<j$, in \S\ with the same signature such that the middle sets $M(A_i, B_i)$, $M(A_j, B_j)$ differ in at least one vertex.

  \newcommand{\glueg}{\ensuremath{G_{ij}}}
  Let~$\glueg \coloneqq G\langle A_i \rangle \glue
  G\langle B_j \rangle$,
  wherein $G\langle A_i \rangle$ is \bdrd{\beta_i} and $G\langle B_j
  \rangle$ is \bdrd{\beta_j}, and let~$W' \coloneqq
  V(\glueg)$.
  Note that $|W'|<|W|$
  since the middle sets of the two edge bipartitions differ in at least one vertex and since $A_i \subsetneq A_j$.

  We prove that \glueg\ is a representative support for $\Hyp$, that
  is, each vertex $\Ver \setminus W'$ is covered in $\Hyp$ by some
  vertex in $W'$ and that \glueg\ is a support for $\Hyp|_{W'}$.
  Herein, for the sake of simpler notation, when referring to the covering condition of representative supports, we assume that each vertex in \glueg\ that results from identifying two vertices in $G$ is equal to an arbitrary one of the two identified vertices.
  Since $\glueg$ is \loutp\ by \cref{sepseq},
  \cref{sepseq:glueable}, the existence of \glueg\ contradicts the choice of~$G$ according
  to the minimum number of vertices, thus proving the lemma.

  To prove that each vertex $\Ver \setminus W'$ is covered by some vertex in $W'$, we show that $\{[u]_\twnrl \mid u \in \Ver\} = \{[u]_\twnrl \mid u \in W'\}$.
  (Hence, every vertex in $\Ver \setminus W'$ has a twin in $W'$; observe that twins cover each other.)
  Since $G = (W, E)$ is a representative support, $\{[u]_\twnrl \mid u \in V\} = \{[u]_\twnrl \mid u \in W\}$.
  Furthermore, by the definition of signature, we have $\{[u]_\twnrl \mid u \in \bigcup A_i\} = \{[u]_\twnrl \mid u \in \bigcup A_j\}$.
  Thus, for each vertex~$u \in W \setminus W'$,
  there is a vertex $v \in W'$ with $[u]_\twnrl = [v]_\twnrl$, meaning that, indeed, $\{[u]_\twnrl \mid u \in V\} = \{[u]_\twnrl \mid u \in W'\}$.

  To show that \glueg\ is a representative support it remains to show that it is a support for $\Hyp|_{W'}$, that is, each hyperedge~$F'$ of~$\Hyp|_{W'}$ induces a connected graph~$\glueg[F']$. Let~$F$ be a hyperedge of~$\Hyp$ such that~$F\cap W'=F'$. Observe that such a hyperedge~$F$ exists and that $G[F \cap W]$~is connected since~$G$ is a representative support of~$\Hyp$. Denote by $S_k$ the middle set~$M(A_k, B_k)$ of $(A_k, B_k)$ in $G$ for $k \in \{i, j\}$ and by $S$ the middle set $M(A_i,B_j)$ of $(A_i, B_j)$ in \glueg.
  Note that the set~$S$ is obtained by identifying the vertices of~$S_i$ with~$S_j$ in the construction of $G_{ij}$. %

  To show that~$\glueg[F']$ is connected, consider first the case that $F \cap (S_i \cup S_j) = \emptyset$. Since each vertex in $V \setminus W'$ is covered
  by a vertex in~$W'$,
  we have that all vertices in~$F$ are contained in either~$G\langle A_i \rangle$ or $G \langle B_j \rangle$ along with all edges of $G[F]$.
  All these edges are also present in \glueg\ whence $\glueg[F']$ is connected.

  Now consider the case that $F \cap (S_i \cup S_j) \neq \emptyset$. Since $S_i$ and $S_j$ are separators in~$G$, each vertex in~$F \setminus (S_i \cup S_j)$ is connected in~$G[F]$ to some vertex in $S_i$ or~$S_j$ via a path with internal vertices in~$F \setminus (S_i \cup S_j)$. We consider the connectivity relation of their corresponding vertices in $S$. To this end, for a graph $H$ and $T \subseteq V(H)$ use $\gamma(T, H)$ for the equivalence relation on~$T$ of connectivity in~$H$. That is, for $u, v \in T$ we have $(u, v) \in \gamma(T, H)$ if $u$ and $v$ are connected in~$H$. Using this terminology, since both $S_i$ and $S_j$ equal $S$ in \glueg, to show that $\glueg[F']$ is connected, it is enough to prove %
  that the transitive closure~$\delta$ of \(\gamma(F' \cap S, \glueg\langle A_i\rangle) \ \cup\ \gamma(F' \cap S, \glueg\langle B_j\rangle)\) contains only one equivalence class.%

  Denote by $\hat{G}$ the graph obtained from $G$ by identifying each $v \in S_i$ with $\beta_j^{-1}(\beta_i(v)) \in S_j$, hence, identifying $S_i$ and $S_j$, resulting in the set~$S$. Relation $\alpha := \gamma(F \cap S, \hat{G})$ has only one equivalence class and, moreover, it is the transitive closure of
  \(\gamma(F \cap S_i, G\langle A_i\rangle)\ \cup\ \gamma(F \cap S, \hat{G}\langle B_i \setminus B_j\rangle)\ \cup\ \gamma(F \cap S_j, G\langle B_j\rangle)\), wherein we identify each $v \in S_i$ with $\beta_j^{-1}(\beta_i(v)) \in S_j$ as above and, thus, $S_i = S_j = S$.
  We have~\(\gamma(F' \cap S, \glueg \langle A_i\rangle) \meq \gamma(F \cap S_i, G\langle A_i\rangle)\) and \(\gamma(F' \cap S, \glueg \langle B_j\rangle) \meq \gamma(F \cap S_j, G\langle B_j\rangle).\) Thus for $\alpha \meq \delta$ it suffices to prove that \(\gamma(F \cap S, \hat{G}\langle B_i \setminus B_j\rangle) \ \subseteq\ \gamma(F' \cap S_j,\glueg\langle B_j\rangle).\) Indeed, the left-hand side $\gamma(F \cap S, \hat{G}\langle B_i \setminus B_j\rangle)$ is contained in $\gamma(F \cap S_i, G \langle B_i\rangle)$: Let $(\T, \phi, \C)$ be the signature of~$(A_i, B_i, \beta_i)$ and $(A_j, B_j, \beta_j)$ and $(F, \gamma_F) \in \C$. Note that $ \gamma(F \cap S_i, G \langle B_i\rangle) = \gamma_F = \gamma(F \cap S_j, G \langle B_j\rangle)$ where we abuse notation and set $u = \beta_i(u)$ for $u \in S_i$ and $v = \beta_j(v)$ for $v \in S_j$. Hence, $\gamma(F \cap S, \hat{G}\langle B_i \setminus B_j\rangle) \subseteq \gamma(F \cap S_j, G \langle B_j \rangle) = \gamma(F' \cap S_j, G \langle B_j \rangle) = \gamma(F' \cap S_j, \glueg \langle B_i\rangle)$. Thus, indeed, $\delta = \alpha$, that is, $F'$ is connected.
\end{proof}

\noindent We now use the upper bound on the number of vertices in representative supports to get rid of superfluous twins. %
First, we show that
representative supports can be extended to obtain a support.
\begin{lemma}\label{lem:representative-hypergraph}
  Let~$G=(W,E)$ be a representative \loutp\ support for a hypergraph~$\Hyp =
  (\Ver,\Ed)$. Then $\Hyp$ has an \loutp\ support in which all vertices
  of~$\Ver\setminus W$ have degree one.
\end{lemma}
\begin{proof}
  Let $G'$~be the graph obtained from~$G$ by making each vertex~$v$ of~$\Ver\setminus W$ a degree-one neighbor of a vertex in~$W$ that covers~$v$ (such a vertex exists by the definition of representative support). Clearly, the resulting graph is planar. It is also \loutp, which can be seen by adapting an \loutp\ embedding of $G$ for~$G'$: If the neighbor~$v$ of a new degree-one vertex~$u$ is in~$L_1$, then place $u$ in the outer face. If $v \in L_i$, $i > 1$, then place $u$ in a face which is incident with~$v$ and a vertex in~$L_{i-1}$ (such a face exists by the definition of~$L_i$).

  It remains to show that~$G'$ is a support for~$\Hyp$. Consider a hyperedge~$F\in \Ed$. Since~$G$ is a representative support for~$\Hyp$, we have that~$F\cap W$ is nonempty and that $G[F\cap W]$ is connected. In~$G'$, each vertex~$u\in F\setminus W$ is adjacent to some vertex~$v\in W$ that covers~$u$. Hence~$v \in F$. Thus, $G'[F]$ is connected as $G'[F\cap W]$ is connected and all vertices in~$F\setminus W$ are neighbors of a vertex in~$F\cap W$.
\end{proof}

We now use \cref{lem:representative-hypergraph} to show that, if there is a twin class that contains more vertices than a small representative support, then we can safely remove one vertex from this twin class.
\begin{lemma}\label{lem:rule}
  Let $\ell \in \univ{N}$, let~$\Hyp$ be a hypergraph, and let~$v\in \Ver(\Hyp)$ be a
  vertex such that $|[v]_{\twnrl}|\ge \ell$. If~$\Hyp$ has a representative
  \loutp\ support with less than~$\ell$ vertices, then~$\Hyp-v$ has an \loutp\ support.
\end{lemma}
\begin{proof}
  Let~$G=(W,E)$ be a representative \loutp\ support for~$\Hyp$ such that $|W|<\ell$. Then at least one vertex of~$[v]_{\twnrl}$ is not in~$W$ and we can assume that this vertex is~$v$ without loss of generality. Thus, $\Hyp$ has an \loutp\ support~$G'$ in which~$v$ has degree one by~\cref{lem:representative-hypergraph}. The graph~$G'-v$ is an \loutp\ support for~$\Hyp-v$: For each hyperedge~$F$ in~$\Hyp-v$, we have that~$G'[F\setminus \{v\}]$ is connected because~$v$ is not a cut-vertex in~$G'[F]$ (since it has degree one).
\end{proof}
\noindent Now we combine the observations above with the fact that
there are small \loutp\ supports to prove that \cref{ps:rr} is correct. %

\begin{proof}[Correctness proof for \cref{ps:rr}]
  Consider an instance~$\Hyp=(\Ver,\Ed)$ of \PS{}
  to which \cref{ps:rr} is applicable and let~$v\in \Ver$ be a vertex to be removed, that is, $v$ is contained in a twin class of size more
  than~$\psi(m, r)$. By~\cref{lem:supp-size}, if~$\Hyp$ has an \loutp\
  support, then it has a representative \loutp\ support with at most
  $\psi(m, r)$ vertices. By~\cref{lem:rule}, this implies that~$\Hyp-v$
  has an \loutp\ support. Moreover, if~$\Hyp-v$ has an \loutp\
  support, then this \loutp\ support is a representative \loutp\
  support for~$\Hyp$. By~\cref{lem:representative-hypergraph}, this
  implies that~$\Hyp$ has an \loutp\ support. Therefore, $\Hyp$~and~$\Hyp-v$ are equivalent instances of \PS, that is, $\Hyp$ has an \loutp\ support if and only if $\Hyp - v$ has,
  and~$v$ can be safely removed
  from~$\Hyp$.
\end{proof}

\section{A sequence of gluable edge bipartitions}\label{sec:sepseq}

In this section, we prove \cref{sepseq}.
For convenience, it is restated below.
Recall also from \cref{sec:application} the intuitive description of the theorem statement, the definition of middle set for an edge %
bipartition, the subgraph~$G\langle A\rangle$ induced by an edge set, and from \cref{sec:prelim} the definitions of graph gluing, boundary, and boundary labeling.

\renewcommand{\thetheorem}{\ref{sepseq}}
\begin{theorem}
  \sepseqthmstatement
\end{theorem}
\addtocounter{theorem}{-1}

The proof relies crucially on \scbd s~\cite{DorPBF10,MP15}.
Recall the corresponding definitions from \cref{sec:sparse-planar-prelim}.

\paragraph{Outline of the proof of \cref{sepseq}.} We first transform the planar embedding of~$G$ into an embedding in the sphere.
Using the fact that $r$-outerplanar graphs have branchwidth at most $2r$~\cite{Biedl15}, we may apply \cref{scbd}, from which we obtain a \scbd{} for $G$ of width at most~$2r$ (recall that the width of a branch decomposition is the size of its largest middle set).
The edge bipartitions in \cref{sepseq} are defined based on the edges in a longest path in the decomposition tree corresponding to the \scbd. The longest path in the decomposition tree has length at least $2\log(n)$, and the edges on this path define a sequence of edge bipartitions, a supersequence of the one in \cref{sepseq}. We define a labeling for each bipartition, which is a string containing \signumbits\ bits, that determines the pairs of edge bipartitions that can be glued so to obtain an \loutp\ graph. The sequence in \cref{sepseq} is then obtained from those bipartitions that have the same labeling. The sphere-cut property of the branch decomposition gives one noose in the sphere for each edge bipartition in the sequence, such that it separates the parts in the edge bipartition from one another. The nooses of the \scbd\ will be crucial in the proof of \cref{sepseq:glueable} in \cref{sepseq}, that is, the \loutp{ity} of the glued graphs.

Let us give some more details concerning the \loutp{ity} of the glued graphs. After sanitizing the nooses, we can assume that they separate the sphere into \emph{left} disks and \emph{right} disks in such a way that each left disk contains all left disks with smaller indices. Hence, for each pair of nooses, we can cut out a left disk and a right disk, and glue them along their corresponding nooses such that we again get a sphere. Alongside the sphere, we get a graph embedded in it that corresponds to the left and right sides of the separators induced by the nooses. It then remains to make the gluing so that the graph remains \loutp, that is, it results in a graph embedded without edge crossings such that each vertex has a face path of length at most~$r$ to the outer face. For this we define a labeling for each edge bipartition and we keep only the largest subsequence of edge bipartitions that have the same labeling.

The labelings roughly work as follows. We want to use the labelings to ensure that the layer of each vertex in $G\langle A_i \rangle \glue G \langle B_j \rangle$ does not increase in comparison to~$G$. For this, for each face touched by the nooses that correspond to $(A_i, B_i)$ and $(A_j, B_j)$, we note in the labelings how far it is away from the outer face (or, rather, the face in the sphere corresponding to the outer face in the plane), and we note for each pair of faces touched by the noose how far they are away from each other. Then, if two edge bipartitions have the same labeling, each vertex in the glued graph will be at most as far away from the faces touched by the noose and, hence, at most as far away from the outer face.

As we will see below, each edge-bipartition labeling can be encoded in $\signumbits$ bits. Thus, out of the $2\log(n)$ edge bipartitions that we obtain from the longest path in the decomposition tree, there are at least $\log(n)/\signumhalved$ edge bipartitions with the same labeling.

The rest of this section is dedicated to the formal proof of \cref{sepseq}.
\newcommand{\glueg}{\ensuremath{G_{ij}}}%
\begin{proof}[Proof of \cref{sepseq}]%
  In the following, fix an arbitrary \loutp\ embedding of~$G$.

  \proofparagraphns{An initial sequence $\T$ of edge bipartitions.} Consider the canonical embedding of $G$ into a sphere~\sph\ that we
  obtain by taking a circle that encloses but does not intersect~$G$
  and identifying all points in the unbounded region of the plane
  that is separated off by this circle. Since $G$ is \loutp, it has
  branchwidth at most~$2r$~\cite[Lemma~3]{Biedl15}. By \cref{scbd}, there is a
  \scbd~$(T, \lambda)$ for $G$ of width at most~$2r$. We define the
  sequence in \cref{sepseq} based on $(T, \lambda)$.

  Consider a longest path~$P$ in $T$.
  Note that the endpoints of $P$ are leaves of~$T$.
  Denote by $e_1$ that edge of~$G$
  that is the preimage of the first vertex of~$P$ under
  the mapping~$\lambda$ of leaves of $T$ to edges of~$G$.
  Since each edge in $T$ induces a bipartition of the edges
  in $G$, so does each edge on $P$. Define the sequence $\T \coloneqq
  ((C_i, D_i))_{i = 1}^t$, where $(C_i, D_i)$ is the bipartition of
  $E(G)$ induced by the~$i$th~edge on~$P$ such that $e_1 \in C_i$. We
  have $C_i \subsetneq C_{i + 1}$ and $D_i \supsetneq D_{i + 1}$
  because $T$ is a ternary tree and $\lambda$ is a bijection. We later
  need a lower bound on the length of~$\T$. For this, observe that $P$
  contains at least $2\log(n)$ edges, because $G$ contains at least~$n$
  edges (there are no vertices of degree one) and $T$ is a ternary
  tree. Hence, sequence~\T also has at least $2\log(n)$ entries. The
  sequence in \cref{sepseq} is defined based on a subsequence of~\T.

\proofparagraph{Obtaining a sequence of noncrossing nooses.}
To define the desired subsequence of \T, we choose one noose~$\geom{N}_i$
for each $(C_i, D_i) \in \T$ such that the resulting sequence of
nooses has the following property. Denote by $\geom{C}_i, \geom{D}_i$
the open disks in which $\geom{N}_i$ separates $\geom{S}$ such that
$C_i \subseteq \geom{C}_i$ and $D_i \subseteq \geom{D}_i$. Then, it
holds that for any two $i, j$ with $i < j$ we have $\geom{C}_i
\subsetneq \geom{C}_{j}$ and $\geom{D}_i \supsetneq \geom{D}_{j}$. If this is the case, then we
say that the nooses $\geom{N}_i$ and $\geom{N}_j$ are
\emph{noncrossing}. Otherwise, that is, if $\geom{C}_i \cap \geom{D}_{j} \neq \emptyset$ and, equivalently, $\geom{D}_i \cap \geom{C}_{j} \neq \emptyset$, then we say that $\geom{N}_i$ and $\geom{N}_j$ \emph{cross} each other.
See
\cref{fig:orange} for examples.
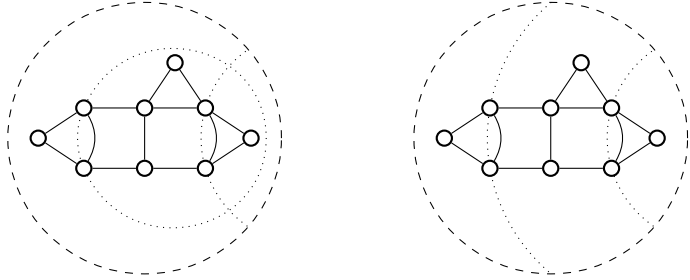
\begin{figure}[t]
  \centering \hfill \begin{tikzpicture}[scale=.4,nodes={vert}, every edge/.style={draw, semithick}]

\draw[dashed]  (0,0) ellipse (4.5 and 4.5);

\node (v1) at (-3.5,0) {};
\node (v2) at (-2,1) {};
\node (v3) at (-2,-1) {};
\node (v4) at (0,1) {};
\node (v5) at (0,-1) {};
\node (v7) at (2,1) {};
\node (v8) at (2,-1) {};
\node (v9) at (3.5,0) {};
\node (v6) at (1,2.5) {};

\begin{pgfonlayer}{background}
  \clip  (0,0) ellipse (4.5 and 4.5);
  \mppath[dotted, thick]{(-2,1) .. (-2,-1) .. (4,0) .. cycle};
  \mppath[dotted, thick]{(3.5,3) .. (2,1) .. (2,-1) .. (3.5,-3)};
\end{pgfonlayer}

\draw  (v1) edge (v2);
\draw  (v1) edge (v3);
\draw  (v3) edge[bend right] (v2);
\draw  (v4) edge (v2);
\draw  (v3) edge (v5);
\draw  (v4) edge (v5);
\draw  (v4) edge (v6);
\draw  (v4) edge (v7);
\draw  (v6) edge (v7);
\draw  (v8) edge (v5);
\draw  (v7) edge[bend left] (v8);
\draw  (v8) edge (v9);
\draw  (v7) edge (v9);

\end{tikzpicture}
  \hfill \begin{tikzpicture}[scale=.4,nodes={vert},every edge/.style={draw, semithick}]

\draw[dashed]  (0,0) ellipse (4.5 and 4.5);

\node (v1) at (-3.5,0) {};
\node (v2) at (-2,1) {};
\node (v3) at (-2,-1) {};
\node (v4) at (0,1) {};
\node (v5) at (0,-1) {};
\node (v7) at (2,1) {};
\node (v8) at (2,-1) {};
\node (v9) at (3.5,0) {};
\node (v6) at (1,2.5) {};

\begin{pgfonlayer}{background}
  \clip  (0,0) ellipse (4.5 and 4.5);
  \mppath[dotted, thick]{(0.7,5) .. (-2,1) .. (-2,-1) .. (.7,-5)};
  \mppath[dotted, thick]{(3.5,3) .. (2,1) .. (2,-1) .. (3.5,-3)};
\end{pgfonlayer}

\draw  (v1) edge (v2);
\draw  (v1) edge (v3);
\draw  (v3) edge[bend right] (v2);
\draw  (v4) edge (v2);
\draw  (v3) edge (v5);
\draw  (v4) edge (v5);
\draw  (v4) edge (v6);
\draw  (v4) edge (v7);
\draw  (v6) edge (v7);
\draw  (v8) edge (v5);
\draw  (v7) edge[bend left] (v8);
\draw  (v8) edge (v9);
\draw  (v7) edge (v9);

\end{tikzpicture}
  \hfill\mbox{}
  \caption{A graph embedded in the sphere and two crossing nooses
    (dotted, left) and two noncrossing nooses (dotted, right). We
    projected the sphere into the plane by replacing a point in the
    sphere with a circle (dashed) and drawing all remaining points
    inside this circle. Both pairs of nooses represent the same edge
    bipartitions. Note that the two nooses on the right share a point
    on the sphere.}
  \label{fig:orange}
\end{figure}

To see that we can choose the nooses in this way, first choose them
arbitrarily and then consider two crossing nooses $\geom{N}_i,
\geom{N}_j$, $i < j$, that is, $\geom{C}_{i} \cap \geom{D}_{j} \neq
\emptyset$. We define a noose~$\tilde{\geom{N}}_i$ which we obtain
from~$\geom{N}_i$ by replacing each maximal subsegment contained in
$\geom{D}_j$ by the corresponding subsegment of $\geom{N}_j$ which is
contained in~$\geom{C}_i$.
There is no edge of $G$ contained in
$\geom{C}_{i} \cap \geom{D}_{j}$ because such an edge then would also
be in $C_{i} \cap D_{j} \subseteq C_i \cap D_i$, a contradiction to
the fact that $C_i, D_i$ is a bipartition of $E(G)$. Hence,
noose~$\tilde{\geom{N}}_i$ separates $\geom{S}$ into two open disks
$\tilde{\geom{C}}_i, \tilde{\geom{D}}_i$ such that $C_i =
\tilde{\geom{C}}_i \cap E(G)$ and~$D_i = \tilde{\geom{C}}_i \cap
E(G)$. Thus, $\tilde{\geom{N}}_i$ fulfills the conditions for the
nooses in \scbd{s} and we may choose $\tilde{\geom{N}}_i$ for $(C_i,
D_i)$ instead of~$\geom{N}_i$.

Clearly, $\tilde{\geom{N}}_i$ and $\geom{N}_j$ are
noncrossing. Moreover, any noose $\geom{N}_k$, $k > i$, that crosses $\tilde{\geom{N}}_i$ also crosses~$\geom{N}_i$ because
$\tilde{\geom{C}}_i \subseteq \geom{C}_i$. Thus, by replacing
$\geom{N}_i$ with $\tilde{\geom{N}}_i$, the number of pairs of
crossing nooses with indices at least~$i$ is strictly decreased. This
means that after a finite number of such replacements we reach a
sequence of pairwise noncrossing nooses.

\proofparagraph{Labelings that allow gluing.}
Based on the sequence $\T$ of edge bipartitions of~$G$ and the nooses
we have fixed above for each edge bipartition, we now define a tuple,
the labeling, for each edge bipartition that can be encoded using
$\signumbits$ bits and that has the property that, if two edge
bipartitions have the same labeling, then the corresponding graphs
can be glued in a way that results in an \loutp\ graph, as stated in
\cref{sepseq}.

 We need some notation and definitions. Denote by $\geom{F}$ the face
in the sphere embedding of~$G$ that corresponds to the outer face of
$G$ in the planar embedding. Pick a point~$y \in \geom{F}$ in such a
way that $y$ is not equal to any vertex and not contained in any edge
or noose~$\geom{N}_i$. For every noose~$\geom{N}_i$ we define a
bijection $\beta_i \colon \{1, 2, \ldots, |M(C_i, D_i)|\} \to M(C_i, D_i)$
corresponding to the order in which the vertices in $M(C_i, D_i)$
appear in a clockwise (in the plane embedding) traversal of $\geom{N}_i$ that starts in an arbitrary
point.
We furthermore define a mapping~$\gamma_i$
from $\{1, 2, \ldots, |M(C_i, D_i)|\}$ to the set of faces touched
by $\geom{N}_i$ according to their occurrence in the traversal of $\geom{N}_i$
above.
More precisely, if face~$\geom{G}$ occurs in the traversal of
$\geom{N}_i$ between vertex~$\beta_i(\ell)$ and $\beta_i(\ell +
1)$ (wherein we put~$\beta_i(|M(C_i, D_i)| + 1)$ equal to~$\beta_i(1)$), then $\gamma_i(\ell) = \geom{G}$.
Finally, say that a face path $P$ is
\emph{contained} in a closed disk $\geom{E}$ if each vertex in $P$ is
contained in~$\geom{E}$.

The
\emph{labeling} of $(C_i, D_i)$ is a tuple~$(b, \Upsilon_1, \Upsilon_2)$ defined as follows.
\begin{compactitem}
\item $b = 1$ if $y \in \geom{C}_i$ and $b = 0$ otherwise\label[piece]{siginf:x} (this encodes which of the disks $\geom{C}_i$ or $\geom{D}_i$
  contains the ``left'' side of the graph).
\item $\Upsilon_1$ is the function defined as follows. Let
  $\beta, \gamma, \geom{C}, \geom{D}$ be symbols.
  Function $\Upsilon_1$ maps each triple $(k, \xi, \geom{X})$ such that $k \in \{1, 2, \ldots, |M(C_i, D_i)|\}$, $\xi \in \{\beta, \gamma\}$, and $\geom{X} \in \{\geom{C}, \geom{D}\}$ to the length of a shortest face path that is contained in~$\geom{X}_i \cup \geom{N}_i$ and that runs from $\xi_i(k)$ to~$\geom{F}$. (Herein, $\xi_i$ refers to $\beta_i$ if $\xi = \beta$ for example and analogously for~$\geom{X}$.)\label[piece]{siginf:pathtoouter}
\item $\Upsilon_2$ is the function that maps each quintuple
  $(k_1, k_2, \xi, \psi, \geom{X})$ such that $k_1, k_2 \in \{1, 2, \ldots, |M(C_i, D_i)|\}$, $\xi, \psi \in \{\beta, \gamma\}$, and $\geom{X} \in \{\geom{C}, \geom{D}\}$ to the length of a shortest face path that is contained in~$\geom{X}_i \cup \geom{N}_i$ from $\xi_i(k_1)$ to $\psi_i(k_2)$.
  \label[piece]{siginf:pathininner}
\end{compactitem}
If the paths above do not exist, or the lengths are larger than
$\layers$, then put $\infty$ instead of the length~$\ell$.

\proofparagraph{Definition of the desired edge bipartition sequence.}
Take \[\S \coloneqq ((C_i, D_i, \beta_i))_{i = 1}^s\] where, in a slight
abuse of notation, $((C_i, D_i))_{i = 1}^s$ is the longest subsequence
of $\T$ in which all edge bipartitions~$(C_i, D_i)$ have the same
labeling. Two edge bipartitions (defined via nooses) which have the
same labeling are shown in \cref{fig:orange} and in
\cref{fig:glue}.
We claim that $\S$ fulfills the conditions of
\cref{sepseq}.

\proofparagraph{Length of the sequence.}  To see that the length $s$ of
$\S$ is large enough, recall that sequence \T contains at
least~$2\log(n)$ entries. The longest subsequence of $\T$ with
pairwise equal labelings has length at least~$2\log(n)$ divided by
the number of different labelings $(b, \Upsilon_1, \Upsilon_2)$. It is not hard to see that there
are at most two possibilities for $b$, at most $(r +
1)^{2r \cdot 2 \cdot 2} = (r + 1)^{8r}$ possibilities for~$\Upsilon_1$, and at most $(r + 1)^{2r \cdot 2r \cdot 2
  \cdot 2 \cdot 2} = (r + 1)^{32r^2}$ possibilities for~$\Upsilon_2$,
giving an overall upper bound on the number of different labelings of
\[ 2 \cdot (r + 1)^{8r} \cdot (r + 1)^{32r^2} \meq 2 \cdot (r +
1)^{32r^2 + 8r}.\] Thus $\S$ has length at least
$\log(n)/\signumhalved$.

\proofparagraph{Outerplanarity number of the glued graphs.}
For each $(C_i, D_i)$, $(C_j,
D_j) \in \S$, $i < j$, we have $C_i \subsetneq C_j$ and $D_i
\supsetneq D_j$. Thus to prove \cref{sepseq} it remains to show that
$\glueg \coloneqq G\langle C_i\rangle \glue G \langle D_j \rangle$ is
\loutp. To see this, we first describe how to obtain an \loutp\
embedding for a supergraph~$G'$ of \glueg\ from $G$'s embedding in the
sphere. %
Graph~$G'$ is defined below and is isomorphic to \glueg\ except that
it may contain multiple copies of some edges in \glueg.

Recall that the nooses $\geom{N}_i$ and $\geom{N}_j$ are
noncrossing. Hence the closed disks~$\geom{C}_i \cup \geom{N}_i$ and
$\geom{D}_j \cup \geom{N}_j$ can intersect only in their nooses $\geom{N}_i$
and $\geom{N}_j$. We now consider dislocating these disks from the
sphere, and identifying their boundaries $\geom{N}_i$ and
$\geom{N}_j$, thus creating a new sphere. \Cref{fig:glue} serves as an
example: Consider the left picture as a plane embedding of the sphere. Take $\geom{C}_i$ to be the outer region of the outermost noose and $\geom{D}_j$ to be the the inner region of the innermost noose. Then dislocate $\geom{C}_i$ and $\geom{D}_j$ from this sphere and glueing them on their nooses; the right picture shows a plane embedding of a sphere that can be obtained in this way.
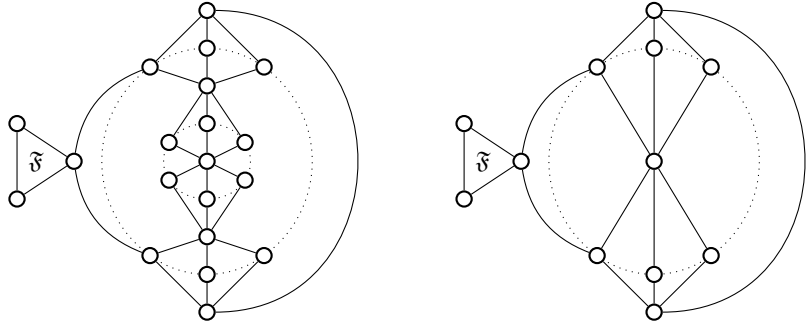
\begin{figure}[t]%
  \centering\hfill\begin{tikzpicture}[scale=.49, every edge/.style={draw, semithick}]
\mppath[dotted, thick]{(-1,0.5) .. (-1,-0.5) .. (0,-1) .. (1,-0.5) ..  (1,0.5) ..  (0,1) .. cycle}
\mppath[dotted, thick]{(-1.5,-2.5) .. (0,-3) .. (1.5,-2.5) ..  (1.5,2.5) .. (0,3) .. (-1.5,2.5) .. cycle}

\node[vert] (v15) at (-1,0.5) {};
\node[vert] (v14) at (0,1) {};
\node[vert] (v16) at (1,0.5) {};
\node[vert] (v20) at (-1,-0.5) {};
\node[vert] (v19) at (0,-1) {};
\node[vert] (v18) at (1,-0.5) {};
\node[vert] (v6) at (0,2) {};
\node[vert] (v10) at (0,-2) {};
\node[vert] (v12) at (0,-3) {};
\node[vert] (v9) at (-1.5,-2.5) {};
\node[vert] (v11) at (1.5,-2.5) {};
\node[vert] (v5) at (0,3) {};
\node[vert] (v4) at (-1.5,2.5) {};
\node[vert] (v7) at (1.5,2.5) {};
\node[vert] (v8) at (0,4) {};%
\node[vert] (v13) at (0,-4) {};
\node[vert] (v3) at (-3.5,0) {};
\node[vert] (v1) at (-5,1) {};
\node[vert] (v2) at (-5,-1) {};
\draw  (v1) edge (v2);
\draw  (v1) edge (v3);
\draw  (v2) edge (v3);
\draw  (v3) edge[bend left] (v4);
\draw  (v5) edge (v6);
\draw  (v4) edge (v6);
\draw  (v6) edge (v7);
\draw  (v7) edge (v8);
\draw  (v8) edge (v5);
\draw  (v4) edge (v8);
\draw  (v3) edge[bend right] (v9);
\draw  (v9) edge (v10);
\draw  (v10) edge (v11);
\draw  (v12) edge (v10);
\draw  (v12) edge (v13);
\draw  (v9) edge (v13);
\draw  (v13) edge (v11);
\draw  (v6) edge (v14);
\draw  (v6) edge (v15);
\node[vert] (v17) at (0,0) {};
\draw  (v16) edge (v6);
\draw  (v17) edge (v15);
\draw  (v17) edge (v14);
\draw  (v17) edge (v16);
\draw  (v17) edge (v18);
\draw  (v19) edge (v17);
\draw  (v20) edge (v17);
\draw  (v19) edge (v10);
\draw  (v20) edge (v10);
\draw  (v18) edge (v10);
\draw  (v13) edge[bend right=90,looseness=1.6] (v8);

\node at (-4.5,0) {\geom{F}};
\end{tikzpicture}%
%
  \hfill\begin{tikzpicture}[scale=.49, every edge/.style={draw, semithick}]
\mppath[dotted, thick]{(-1.5,-2.5) .. (0,-3) .. (1.5,-2.5) ..  (1.5,2.5) .. (0,3) .. (-1.5,2.5) .. cycle}

\node[vert] (v12) at (0,-3) {};
\node[vert] (v9) at (-1.5,-2.5) {};
\node[vert] (v11) at (1.5,-2.5) {};
\node[vert] (v5) at (0,3) {};
\node[vert] (v4) at (-1.5,2.5) {};
\node[vert] (v7) at (1.5,2.5) {};
\node[vert] (v8) at (0,4) {};%
\node[vert] (v13) at (0,-4) {};
\node[vert] (v3) at (-3.5,0) {};
\node[vert] (v1) at (-5,1) {};
\node[vert] (v2) at (-5,-1) {};
\draw  (v1) edge (v2);
\draw  (v1) edge (v3);
\draw  (v2) edge (v3);
\draw  (v3) edge[bend left] (v4);
\draw  (v7) edge (v8);
\draw  (v8) edge (v5);
\draw  (v4) edge (v8);
\draw  (v3) edge[bend right] (v9);
\draw  (v12) edge (v13);
\draw  (v9) edge (v13);
\draw  (v13) edge (v11);
\node[vert] (v17) at (0,0) {};
\draw  (v13) edge[bend right=90,looseness=1.6] (v8);

\node at (-4.5,0) {\geom{F}};
\draw  (v17) edge (v4);
\draw  (v5) edge (v17);
\draw  (v7) edge (v17);
\draw  (v17) edge (v11);
\draw  (v12) edge (v17);
\draw  (v17) edge (v9);
\end{tikzpicture}%
  \hfill\mbox{}
  \caption{Left: A graph embedded in a subdisk of the sphere which has been projected onto the plane. We show two nooses (dotted) that induce edge bipartitions. The labelings of the two edge bipartitions are the same if we assume that both left sides (the $C_i$'s) of the bipartitions contain the outermost edges in the drawing and if we furthermore assume that the corresponding mappings~$\beta_i$'s are the clockwise orderings of the vertices on the noose starting with the topmost vertex.
    Right: The graph resulting from gluing along the two nooses.}%
  \label{fig:glue}%
\end{figure}%

Observe that there is a homeomorphism~$\phi \colon \geom{C}_i \cup \geom{N}_i
\to \geom{C}_j \cup \geom{N}_j$ since both point sets are closed discs.
Recall that, since $(C_i, D_i)$ and $(C_j, D_j)$ have the same labeling, we have that the vertex~$y \in \geom{F}$ which we fixed above is either both in $\geom{C}_i$ and $\geom{C}_j$ or both in $\geom{D}_i$ and $\geom{D}_j$.
Intuitively, the nooses are thus nested inside each other in the plane.
Recall furthermore that the vertices in $M(C_i,D_i)$ and $M(C_j,D_j)$ are
enumerated by $\beta_i$ and $\beta_j$, respectively, according to clockwise
traversals of the corresponding nooses.
Hence, we may choose the homeomorphism~$\phi \colon \geom{C}_i \cup \geom{N}_i
\to \geom{C}_j \cup \geom{N}_j$ such
that it has the following
properties.
\begin{mathenum}
\item For the two traversals of the nooses that define $\beta_i$ and
  $\beta_j$, respectively, we have that the initial points of the
  traversals are mapped onto each other by $\phi$ and, if point $x$ comes
  before point $z$ in the traversal of $\geom{N}_i$ used to define $\beta_i$, then $\phi(z)$ comes
  after $\phi(x)$ in the traversal of
  $\geom{N}_j$ used to define $\beta_j$.\label[property]{respect}
\item For each $k \in \{1, 2, \ldots, |M(C_i, D_i)|\}$ we have
  $\phi(\beta_i(k)) = \beta_j(k)$.\label[property]{identify}
\end{mathenum}
Denote by $G'$ the \splane\ graph induced by the point set $\phi(G
\cap \geom{C}_i) \cup (G \cap \geom{D}_i)$. We claim that from $G'$ we
can derive an \loutp\ embedding of~\glueg.

We first prove that \glueg\ is an edge-induced subgraph of $G'$
without loss of generality: We may assume that $G$ and \glueg\ have
the same vertex set without loss of generality by \cref{identify} of
homeomorphism~$\phi$. Since each edge $e \in C_i$ is contained in
$\geom{C}_i$, it is also present in $\phi(\geom{C}_i)$ and thus in
$G'$. Moreover, each edge $e \in D_j$ is trivially contained in
$\geom{D}_j$, hence, also in $G'$. Thus, we may assume that \glueg\ is
an edge-induced subgraph of $G'$ whence from any \loutp\ embedding of~$G'$ we obtain an \loutp\ embedding of~\glueg.

Graph~$G'$ has a sphere embedding due to the way it was
constructed. We now prove that from this embedding we can obtain an
\loutp\ one. This then finishes the proof of \cref{sepseq}. Note that there is a face
in the sphere embedding of $G'$ that contains $y$ or $\phi(y)$ due to the flag $b$ in the labelings (i.e.\ if $b = 1$ then there is a face containing $\phi(y)$, otherwise there is a face containing~$y$).
We denote this face by
$\geom{F}$. By removing a point contained in the face~$\geom{F}$ from
the sphere, we obtain a topological space homeomorphic to the
plane. Fix a corresponding homeomorphism~$\delta$ and note that,
applying $\delta$ to $G'$, we obtain a planar embedding of $G'$ with
the outer face~$\delta(\geom{F})$. In the following we assume that
$G'$ is embedded in this way and, for the sake of simplicity, denote
$\delta(\geom{F})$ by $\geom{F}$.

To conclude the proof it remains to show that $G'$ is \loutp. Recall that a graph is \loutp\ if and only it has an
embedding in the plane such that each vertex~$v$ has an incident face
with a face path of length at most~$\layers$ to the outer
face~$\geom{F}$. Call such a path \emph{good} with respect to~$v$.

It remains to show that each vertex in $G'$ has a good path.
It suffices to prove this for vertices in~$\geom{C}_i$ whose good paths in~$G$ are not contained in $\geom{C}_i$ and vertices in~$\geom{D}_j$ whose good paths in~$G$ are not contained in~$\geom{D}_j$ as the remaining ones are also present in~$G'$.
Consider a vertex in $\geom{C}_i$ whose good path~$P$ is not contained in~$\geom{C}_i$.
Observe that each subpath of $P$ that is not contained in $\geom{C}_i$ is contained in $\geom{D}_i \cup \geom{N}_i$.
We claim that we can replace every maximal face subpath of $P$ which is contained in~$\geom{D}_i \cup \geom{N}_i$ by a face path contained in $\geom{D}_j \cup \geom{N}_j$ in such a way that the resulting sequence~$P'$ is a face path in $G'$.
Moreover, $P'$ is at most as long as $P$.

Consider a maximal face subpath~$S$ of $P$ which is contained
in~$\geom{D}_i \cup \geom{N}_i$.
Each end of~$S$ is either a vertex in
$M(C_i, D_i)$, or a face. If an end of~$S$ is a face, then it can
either be the outer face~$\geom{F}$ or a face~$\geom{G} \neq \geom{F}$
which is intersected by~$\geom{N}_i$. (Note that not both ends of $S$
can be $\geom{F}$ as $P$ is a shortest path to $\geom{F}$.)

We now use the labelings to show that there is a suitable replacement for~$S$.
Consider the case where one end of $S$ is $\geom{F}$; we will indeed only treat this case explicitly.
The other case is analogous.
Consider the subcase where the other end of $S$ is a vertex~$v$. 
Associate with $S$ the tuple $(k,\beta,\geom{D})$, where $\beta$ and $\geom{D}$ are the corresponding symbols from the labelings.
The first entry, $k$, is an integer equal to $\beta_i^{-1}(v)$.
Let $\Upsilon_1^i$ be the function $\Upsilon_1$ from the labeling of $(C_i, D_i)$ and $\Upsilon_1^j$ the function $\Upsilon_1$ from the labeling of $(C_j, D_j)$.
Since the two labelings are the same, $\Upsilon_1^i(k, \beta, \geom{D}) = \Upsilon_2^j(k, \beta, \geom{D}) \leq \ell$, where $\ell$ is the length of~$S$ (observe that $\ell \leq \lrs$).
Hence, there is a face
path~$S'$ in $\geom{D}_j$ with the ends~$\geom{F}$ and
$\beta_j(k)$.
By \cref{identify} of homeomorphism~$\phi$ we have that $\beta_i(k) = \beta_j(k)$ in $G'$. 
Thus we can replace $S$ by $S'$ in $P$, that is, afterwards (i) the maximal subsegment of $P$ contained in $G'$ is a face path and (ii) the number of maximal segments of $P$ that are not contained in $G'$ strictly decreases.

Now consider the subcase where the end of $S$ that is different from $\geom{F}$ is a face~$\geom{G}$.
Associate with $S$ the tuple $(k,\gamma,\geom{D})$, where $\gamma$ and $\geom{D}$ are the corresponding symbols from the labelings.
The first entry, $k$, is defined as follows.
Draw an arc~$\geom{A}$ contained in $\geom{G}$ between the two vertices that~$P$~visits before and after $\geom{G}$ such that $\geom{A}$ and
$\geom{N}_i$ intersect in only one point.
Call this intersection point~$x$.
Define $k \in \univ{N}$ such that, in the traversal of
$\geom{N}_i$ that defines $\beta_i$,
vertex~$\beta_{i}(k)$ comes
before $x$ and vertex~$\beta_i(k + 1)$ comes after $x$ (where we set $k
+ 1 = 1$ if $k = |M(C_i, D_i)|$).
Let $\Upsilon_1^i$ be the function $\Upsilon_1$ from the labeling of $(C_i, D_i)$ and $\Upsilon_1^j$ the function $\Upsilon_1$ from the labeling of $(C_j, D_j)$.
Since the two labelings are the same, $\Upsilon_1^i(k, \gamma, \geom{D}) = \Upsilon_2^j(k, \gamma, \geom{D}) \leq \ell$, where $\ell$ is the length of~$S$ (observe that~$\ell \leq \lrs$).
Hence, there is a face
path~$S'$ in $\geom{D}_j$ with the ends~$\geom{F}$ and~$\gamma_j(k)$.

We claim that $\gamma_j(k) = \geom{G}$ and $\gamma_i(k)$ describe the same
entities in $G'$.
Consider the face $\geom{H} = \gamma_j(k)$ in~$G$. By definition,
$\geom{G}$~intersects $\geom{N}_i$ in the segment~$\geom{S}_i$ of the
traversal defining $\beta$ between $\beta_i(k)$ and~$\beta_i(k + 1)$. Similarly, $\geom{H}$ intersects $\geom{N}_i$
in the segment $\geom{S}_j$ between $\beta_j(k)$ and
$\beta_j(k + 1)$. In $G'$, face~$\geom{G}$ is represented by
$\phi(\geom{G} \cap (\geom{C}_i \cup \geom{N}_i))$ and face $\geom{H}$
is represented by $\geom{H} \cap (\geom{D}_j \cup \geom{N}_j) =
\geom{H} \cap (\geom{D}_j \cup \geom{N}_i)$. Moreover, segments
$\geom{S}_i$ and $\geom{S}_j$ are identified by homeomorphism $\phi$
because of its \cref{respect}. Hence, $\phi(\geom{G} \cap (\geom{C}_i
\cup \geom{N}_i))$ and $\geom{H} \cap (\geom{D}_j \cup \geom{N}_i)$
are merged into one face in $G'$. Thus, indeed $\gamma_j(k)$ and~$\gamma_i(k)$ describe the same entities in $G'$.
Thus we can replace $S$ by $S'$ in $P$, that is, afterwards (i) the maximal subsegment of $P$ contained in $G'$ is a face path and (ii) the number of maximal segments of $P$ that are not contained in $G'$ strictly decreases.

The proof that we can replace $S$ by a corresponding path $S'$ in $P$
in the case that $S$ does not have $\geom{F}$ as an
end is analogous to the above and omitted. Hence, replacing all
maximal face subpaths of $P$ that are not contained in $\geom{C}_i$,
we obtain a good path in $G'$. Finally, the case that the good path of
a vertex in $\geom{D}_j$ is not contained in $\geom{C}_i$ is symmetric
to the above and also omitted.

Summarizing, since each vertex in $G$ has a good path, so has each
vertex in $G'$, meaning that $G'$ is \loutp. Since \glueg\ is an
edge-induced subgraph of $G'$, also \glueg\ is \loutp. This concludes
the proof of \cref{sepseq}.
\end{proof}

\section{Concluding remarks}\label{sec:concl}
The main contribution of this work is to show that twins are
crucial for instances of \PS\ but the number of crucial twins is upper-bounded in terms of the number~$m$ of hyperedges and the outerplanarity number~$\lrs$ of a support. As a result, we can safely remove noncrucial twins.
More specifically, in linear time we can transform any instance of
\PS\ into an equivalent one whose size is upper-bounded by a function of $m$
and~$r$ only.
In turn, this implies \fpty\ \wrt\ $m + r$. It is fair to say that
due to the strong exponential growth in $m$ and~$r$ this result is mainly of
classification nature.  Improved bounds (perhaps
based on further data reduction rules) are highly desirable
for practical applications.

Two further directions for future research are as follows.
First, above we only showed how to shrink the size of the input instance. We also need an efficient algorithm to construct an \loutp\ support for such an instance. A naive algorithm for this task has running time~$n^{\Oh(n)}$: Since every planar graph on~$n$ vertices has~$\Oh(n)$ edges, we may enumerate all planar graphs in~$n^{\Oh(n)}$~time by considering all possible endpoints for each edge. For each enumerated graph, we then test whether one of them is an \loutp\ support. Can we improve over this brute-force algorithm?  %

Second, it is interesting to gear the parameters under consideration more towards practice. In \cref{sec:application} above we attached signatures to each edge bipartition in a sequence of edge bipartitions of a support and we could reduce our input only if there were sufficiently many edge bipartitions with the same signature. This signature contained, among other information, the twin class of each vertex of the separator induced by the edge bipartition. Clearly, if all of these at least $2^{m\lrs}$ different types of signatures are present, then this will lead to an illegible drawing of the hypergraph (and still, in absence of better upper bounds, we cannot shrink our input). It seems thus worthwhile to contemplate parameters that capture legibility of the hypergraph drawing by restricting further the number of possible signatures. 

Finally, an obvious open question is whether finding a \emph{planar} support is (linear-time) \fpt\ with respect to the number~$m$ of hyperedges only. A promising direction might be to show that there is a planar representative support (as in \cref{def:representative}) which has treewidth upper-bounded by a function of~$m$. From this, we would get a sequence of gluable subgraphs similarly to the one we have used here, amenable to the same approach as in \cref{sec:application}.

\paragraph{Acknowledgments.} We thank an anonymous referee who pointed out the ideas for the proof in \cref{sec:nonuni-fpt}.

\paragraph{Funding.}
 \marginpar{\vspace{2cm}
   \includegraphics[width=35px]{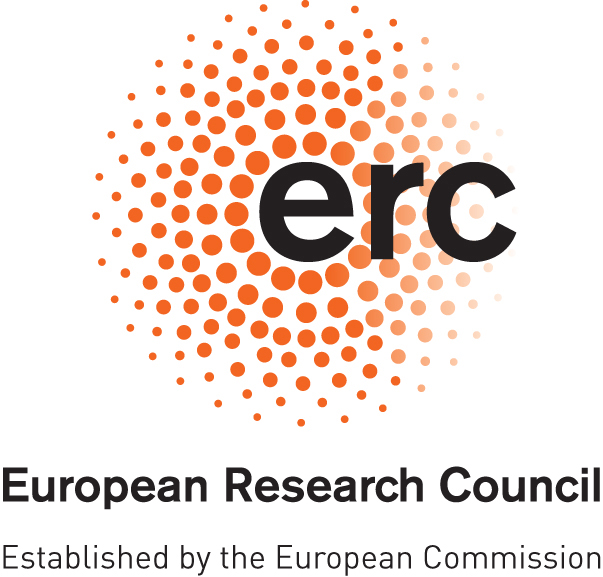}}\marginpar{\includegraphics[width=35px]{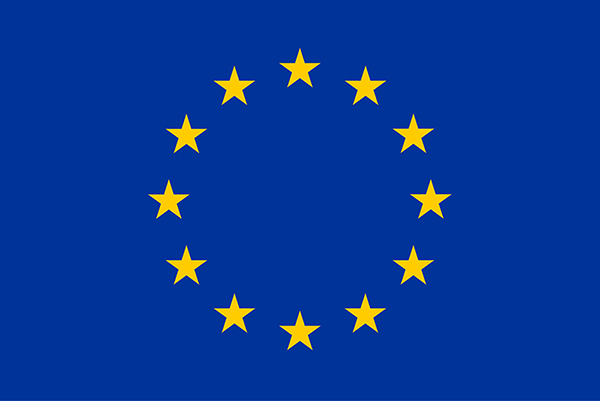}}
Significant parts of the work were done while all authors were with
TU~Berlin and while RvB, IK, and MS were supported by the DFG project
NI~369/12.  %
CK was supported by DFG project KO~3669/4-1.
MS was supported by the People Programme (Marie Curie Actions) of the European Union's Seventh Framework Programme (FP7/2007-2013) under REA grant agreement number {631163.11}, by the Israel Science Foundation (grant no.\ 551145/14), and by the European Research Council (ERC) under the
    European Union’s Horizon 2020 research and innovation programme
    under grant agreement number~714704, and by the Alexander von Humboldt Foundation.
    Some work of MS was done while with Ben-Gurion University of the Negev and with University of Warsaw.

\printbibliography

\end{document}